\documentclass[twocolumn]{article}
\usepackage{amssymb}
\usepackage{graphicx}
\usepackage{subfig}
\usepackage{paralist}
\usepackage{url}
\usepackage{comment}
\usepackage{proof}
\usepackage[draft]{hyperref}

\usepackage[T1]{fontenc}
\usepackage[scaled=0.80]{beramono}
\usepackage[final]{listings}
\usepackage{multirow}
\usepackage[obeyDraft,textsize=tiny]{todonotes}
\usepackage{listings}
\usepackage{colortbl}
\usepackage{pgfplots}
\usepackage{pgfplotstable}
\usepackage{booktabs}
\usepackage{enumitem}
\usepackage{stfloats}
\usepackage{amssymb}
\usepackage{graphicx}
\usepackage{subfig}
\usepackage{paralist}
\usepackage{url}
\usepackage{proof}
\usepackage{hyperref}

\usepackage[T1]{fontenc}
\usepackage[final]{listings}
\usepackage{colortbl}
\usepackage{booktabs}
\usepackage{enumitem}
\usepackage[euler]{textgreek}

\usepackage{algorithm}
\usepackage{algorithmic}
\usepackage{multicol}

\usepackage{siunitx}

\usepackage{multirow}
\usepackage{subfig}
\usepackage{paralist}

\usepackage{amsmath,amsthm}

\DeclareMathOperator*{\argmin}{arg\,min}
\newcolumntype{L}[1]{>{\raggedright\let\newline\\\arraybackslash\hspace{0pt}}m{#1}}
\newcolumntype{C}[1]{>{\centering\let\newline\\\arraybackslash\hspace{0pt}}m{#1}}
\newcolumntype{R}[1]{>{\raggedleft\let\newline\\\arraybackslash\hspace{0pt}}m{#1}}

\begin{document}
%
%\SetKwRepeat{Do}{do}{while}

%%%%%%%%%%%%%%% article 

\theoremstyle{definition}
\newtheorem{definition}{Definition}[section]

\theoremstyle{theorem}
\newtheorem{theorem}{Theorem}[section]

%%%%%%%%%%%%%%%%article end

\title{Compacting Frequent Star Patterns in RDF Graphs}

\author{%
\textsc{Farah Karim} \\[1ex] % Your name
\normalsize Leibniz University of Hannover, Germany \\Mirpur University of Science and Technology (MUST), Mirpur-10250 (AJK), Pakistan \\ % Your institution
\normalsize \href{mailto:Karim@l3s.de}{Karim@l3s.de} % Your email address
\and % Uncomment if 2 authors are required, duplicate these 4 lines if more
\textsc{Maria{-}Esther Vidal} \\[1ex] % Second author's name
\normalsize TIB Leibniz Information Centre for Science and Technology, Hannover\\
Leibniz University of Hannover, Germany  \\ % Second author's institution
\normalsize \href{mailto:Maria.Vidal@tib.eu}{Maria.Vidal@tib.eu} % Second author's email address
\and
\textsc{S{\"o}ren Auer} \\[1ex] % Second author's name
\normalsize TIB Leibniz Information Centre for Science and Technology, Hannover\\
Leibniz University of Hannover, Germany  \\ % Second author's institution
\normalsize \href{mailto:Auer@tib.eu}{Auer@tib.eu} % Second author's email address
}

%\author{Farah Karim\and Maria-Esther Vidal\and  S{\"o}ren Auer}
%\institute{Farah Karim \at
%             Leibniz University of Hannover,
%              Welfengarten 1B,  30167 Hannover, Germany\\
%              Mirpur University of Science and Technology (MUST), Mirpur-10250 (AJK), Pakistan\\
 %             \email{karim@l3s.de} }
\maketitle              % typeset the header of the contribution
\begin{abstract}
%WHY 
Knowledge graphs have become a popular formalism for representing entities and their properties using a graph data model, e.g., the Resource Description Framework (RDF).  
An RDF graph comprises entities of the same type connected to objects or other entities using labeled edges annotated with properties. 
%A \emph{frequent} graph pattern in an RDF graph contains a set of entities in a class having the same node values in a group of properties. 
%These graph patterns usually exist in RDF graphs with a high number of duplicates (i.e., edges and nodes), negatively impacting, thus, size and query processing.
RDF graphs usually contain entities that share the same objects in a certain group of properties, i.e., they match star patterns  composed of these properties and objects. In case the number of these entities or properties in these star patterns is large, the size of the RDF graph and query processing are negatively impacted; we refer these star patterns as \emph{frequent star patterns}.
%WHAT
We address the problem of identifying \emph{frequent star patterns} in RDF graphs and devise the concept of \emph{factorized RDF graphs}, which denote compact representations of RDF graphs where the number of frequent star patterns is minimized. 
We also develop computational methods to identify frequent star patterns and generate a \emph{factorized RDF graph}, where \emph{compact RDF molecules} replace frequent star patterns.
A compact RDF molecule of a frequent star pattern denotes an RDF subgraph that instantiates the corresponding star pattern. Instead of having all the entities matching the original frequent star pattern, a \emph{surrogate} entity is added and related to the properties of the frequent star pattern; it is linked to the entities that originally match the frequent star pattern. Since the edges between the entities and the objects in the frequent star pattern are replaced by edges between these entities and the surrogate entity of the compact RDF molecule, the size of the RDF graph is reduced. We evaluate the performance of our factorization techniques on several RDF graph benchmarks and compare with a baseline built on top of \textit{gSpan}, a state-of-the-art algorithm to detect frequent patterns. The outcomes evidence the efficiency of proposed approach and show that  our techniques are able to reduce execution time of the baseline approach in at least three orders of magnitude. Additionally, RDF graph size can be reduced by up to $66.56\%$ while data represented in the original RDF graph is preserved. 
%\keywords{Semantic Web  \and RDF compaction \and Linked Data \and Knowledge graph.}
\end{abstract}


\section{Introduction}

% Wide adaptation of IoT applications 
% IoT data representation as RDF 
% SSN ontology to represent IoT data and to add semantics 
% heterogeniety in RDF physical representations + data redundancy 
% Existing approaches 
% Objective / Goal: logical representations of SSN data not only to remove redundancy but also enhance physical representations 
% Research questions 
% Approach 
% Contributions 
% paper structure 

Knowledge graphs have gained momentum as flexible and expressive structures for representing not only data and knowledge but also actionable insights~\cite{vidaltransforming}; they provide the basis for effective and intelligent applications. Currently, knowledge graphs are utilized in diverse domains e.g., DBpedia~\cite{lehmann2015dbpedia}, Google Knowledge Graph~\cite{singhal2012introducing}, and KnowLife~\cite{ernst2015knowlife}.  The Resource Description Framework (RDF)~\cite{lassila1998resource} has been adopted as a formalism to represent knowledge graphs; in fact, in the Linked Open Data cloud~\cite{bizer2011linked}, there are in 2019 more than 1,200 RDF knowledge graphs available\footnote{https://lod-cloud.net/}. RDF models knowledge in the form of graphs where nodes represent entities; connections between entity nodes are representing RDF triples composed of subject, property, and object. The subjects and objects are represented by nodes, and an edge represents a property that relates a subject with an object. Diverse applications have been developed on top of knowledge graphs~\cite{AuerKPKSV18,Grangel-Gonzalez18,vidaltransforming}. However, the adoption of knowledge graphs as \emph{de facto data} structure of real-world applications demands efficient representations and scalable techniques for creating, managing, and answering queries over knowledge graphs. Thus, efficient  graph representations of real-world scenarios are still demanded to enhance and facilitate the development of applications over knowledge graphs.

In real-world applications, a group of entities can share the same values in a set of features. For example, several sensor observations can sense the same temperature, in a given timestamp and city.  This situation can be represented in an RDF graph with four triples per sensor observation $o_i$, i.e., $(o_i\; temperature\; t)$,  $(o_i\; unit\;  u)$, $(o_i\; timestamp\; ts)$, and $(o_i\; gps\_coordinates\; gc)$. All the resources representing these sensor observations match the variable $?o$ in the star pattern (SGP) composed by the conjunction of the following triple patterns $(?o\; tem-$ $perature\; t)$ $(?o\; unit\; u)$, $(?o \;timestamp\; ts)$, and $(?o\;$ $gps\_coordinates\; gc)$~\cite{prud2011sparql}. In case the star patterns are instantiated with many entities, a large number of RDF triples will have the same properties and objects and the corresponding star pattern will be repeatedly instantiated; we name these star patterns \emph{frequent star patterns}. Although RDF triples that instantiate a frequent star pattern correctly model the real world, the size of the knowledge graph as well as the efficiency of the tasks of management and processing, can be negatively affected whenever a large number of triples of frequent star patterns populate the knowledge graph. Since frequent star patterns are very common in real world knowledge graphs, techniques are required to enable both the efficient representation of the knowledge encoded in these star patterns, as well as the processing and traversal of the represented knowledge. 

The Database and Semantic Web communities have addressed the problem of representing relational and graph data models; they have proposed a variety of representation methods and data structures that take into account the main features of a relational or graph model with the aim of speeding up relation and graph based analytics~\cite{abadi2006integrating,allen2019understanding,alvarez2011compressed,fernandez2013binary,joshi2013logical,karim2017large,meier2008towards,pichler2010redundancy,zukowski2006super}. Compression techniques~\cite{abadi2006integrating,zukowski2006super} over the column-oriented databases~\cite{boncz2005monetdb,stonebraker2005c}, use the decomposition storage model~\cite{copeland1985decomposition} to maintain data, where each attribute value and a surrogate key, from the conceptual schema, are stored in a binary relation. However, a relation stored using the decomposition storage model cannot easily exploit compression unless surrogate keys are repeated~\cite{copeland1985decomposition}. Further, the decomposition model stores two copies of a binary relation, also the surrogate keys are required to be stored repeatedly for each attribute causing an increase in the storage space requirements. In the context of RDF graph, the scientific community has also actively contributed; approaches like~\cite{alvarez2011compressed,fernandez2013binary,pan2014graph,zhu2018predicate} generate compact binary representations for RDF knowledge graphs. RDF binary compression techniques do not take into account the semantics encoded in knowledge graphs; they require customized engines to perform query processing. Moreover, there have been defined compression approaches for RDF graphs able to exploit semantics encoded in RDF triples. Approaches~\cite{meier2008towards,pichler2010redundancy} are application dependent and require a user to input the compression rules and constraints. Alternatively, compression approaches tailored for ontology properties~\cite{karim2017large} have shown to be effective, but they require prior knowledge of classes and properties involved in repeated graph patterns to generate compact representations. Lastly,  techniques proposed by Joshi et al.~\cite{joshi2013logical} require decompression to access and process the original data, as well as extra processing over the data. Albeit effective in reducing the storage space, existing compression methods add overhead to the process of data management, and particularly, query execution time can be negatively impacted. gSpan~\cite{yan2002gspan} and GRAMI~\cite{elseidy2014grami} are state-of-the-art algorithms that aim to identify frequent patterns. However, only patterns with constants are considered and they are neither able to identify star patterns nor decide \textit{frequentness}. We have built an exhaustive algorithm that resorts to the gSpan enumeration of frequent patterns to identify the frequent star patterns in an RDF knowledge graph; this approach corresponds to the baseline of our empirical evaluation.  

\noindent
\textbf{Our Research Goal:}
We address the problem of \emph{identifying frequent star patterns} in RDF knowledge graphs, where certain properties and their corresponding objects are repeatedly shared by several entities of a type causing unnecessary growth of the knowledge graphs. Our research goal is to minimize the number of \emph{frequent star patterns} in RDF knowledge graphs to generate compact representations without losing any information.
We investigate the following research questions:
\begin{itemize}
    \item What are the criteria that characterize frequent star patterns?
    \item Do compact graph representations impact on the size of knowledge graphs?
\end{itemize}

 \begin{figure*}[t!]
\label{motivationExp}
\centering
     \vspace{0pt}\subfloat[An RDF Graph $G$]{
      \includegraphics[width=0.34\textwidth]{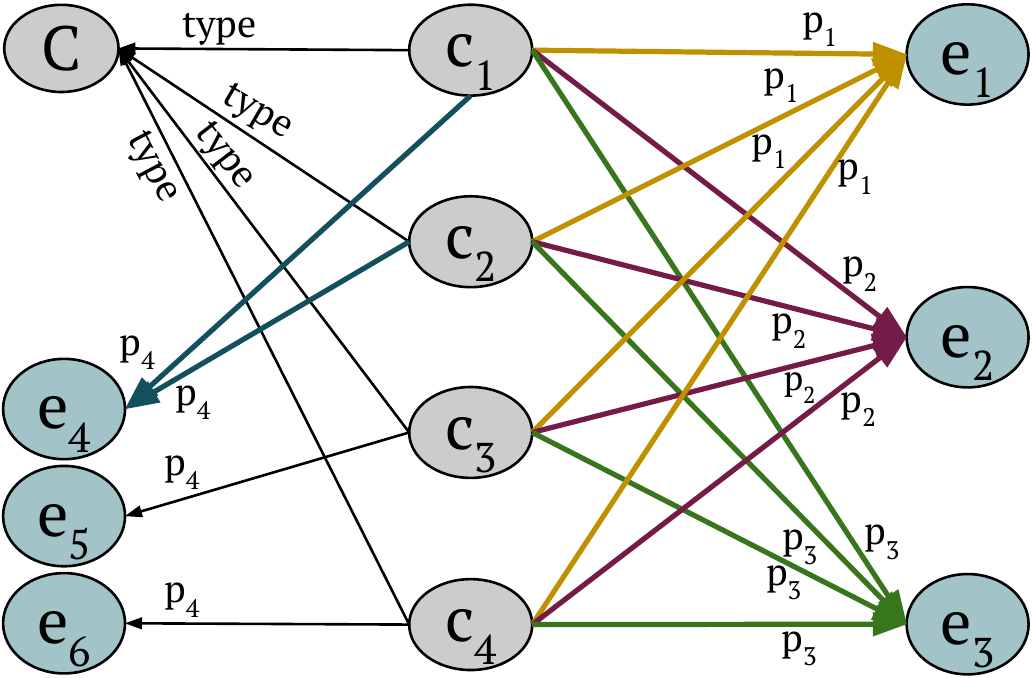}
      \label{fig:duplicates}}
       \vspace{-0.50pt}
       \vspace{0pt}\subfloat[Entities in the Graph Pattern]{
      \includegraphics[width=0.34\textwidth]{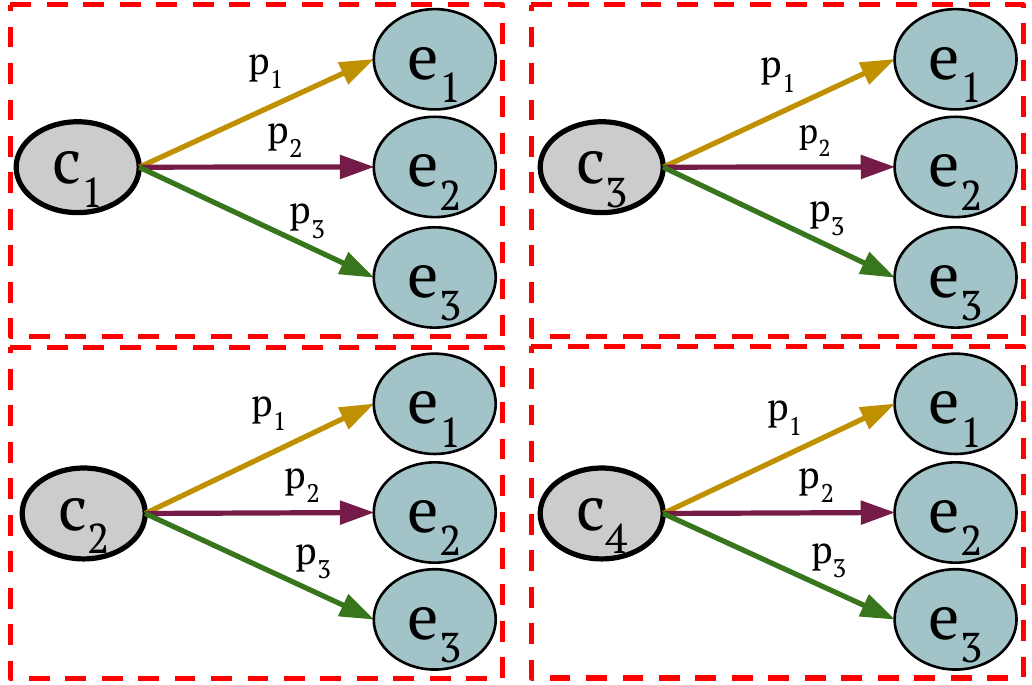}
      \label{fig:fgp}}
       \vspace{-0.50pt}
  \vspace{0pt}\subfloat[A Star Pattern]{
      \includegraphics[width=0.25\textwidth]{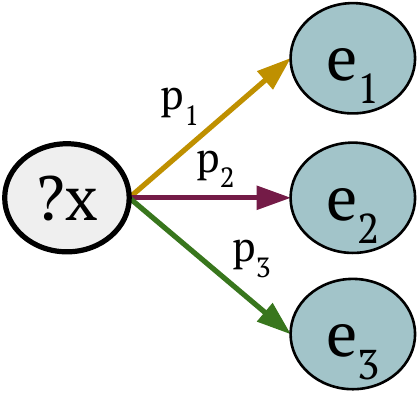}
      \label{fig:redundancy}}
    %\vspace{0pt}\subfloat[No duplicates in RDF graph]{
     % \includegraphics[width=0.48\textwidth]{images/noduplicates.pdf}
      %\label{fig:noduplicates}}
      %Class instances repeatedly connected to the same values using same set of attributes, Set of attributes and values correspond to several class instances without repetition
   \caption{{\bf Motivating Example}. Frequent star pattern. (a) RDF graph with classes, entities, and properties; (b) Entities $c_1$, $c_2$, $c_3$, and $c_4$ are related to $e_1$, $e_2$, and $e_3$ with properties $p_1$, $p_2$, and $p_3$, respectively; (c) A star pattern with subject variable $?x$, respectively, relates $e_1$, $e_2$, and $e_3$ with properties $p_1$, $p_2$, and $p_3$.}
\end{figure*}

\noindent\textbf{Approach:} 
We devise the concept of \emph{factorized RDF graphs}, which corresponds to a compact graph with a minimized number of frequent star patterns. 
Further, we develop computational methods to detect frequent star patterns in RDF graphs and to generate a \emph{factorized RDF graph}. 
These methods are able to identify entities and properties in frequent star patterns in RDF graphs, and generate factorized RDF graphs by representing frequent star patterns with compact RDF molecules. A compact RDF molecule of a frequent star pattern is an RDF subgraph that instantiates the star pattern; a surrogate entity stands for the entities that satisfy the corresponding frequent star pattern. The surrogate entity is linked to the properties and the corresponding objects in the frequent star pattern (see Figure~\ref{fig:elemnetsNumber}). The entities, initially matching the frequent star pattern, are also linked to the surrogate entity of the compact RDF molecule. Compact RDF molecules significantly reduce the size of the RDF graph by replacing labeled edges and entities connected the objects in the frequent star pattern, with edges linking the entities to the surrogate entity of a compact RDF molecule. 
We study the effectiveness of our factorization techniques over the \emph{LinkedSensorData} benchmark~\cite{patni2010linked}; it describes more than 34,000,000 weather observations collected by around 20,000 weather stations in the United States since 2002. Experiments are conducted against three \emph{LinkedSensorData} RDF graphs by gradually increasing the graph size. The observed results evidence that frequent star patterns characterize the best set of properties relating several entities of a class to the same objects in an RDF graph. Moreover, our techniques reduce RDF graphs size by up to $66.56\%$ using properties and classes recommended by the frequent star patterns detection approach. 

\noindent\textbf{Contributions:} we devise computational methods for factorizing RDF graphs. The specific contributions are as follows:
\begin{inparaenum}[\bf i\upshape)]
	\item Criteria for detecting frequent star patterns; 
	\item Factorization techniques compacting frequent star patterns in
RDF graphs. We have presented two algorithms: An exhaustive approach (named E.FSP) searches the space of frequent patterns produced by an algorithm like gSpan, to identify frequent star patterns. Further, G.FSP implements a Greedy meta-heuristics that is able to traverse the space of star patterns and identify the ones that are frequent. Star patterns are traversed in iterations, starting with the star patterns with the largest number of properties. The criteria of frequent star patterns correspond the stop criteria of the algorithm. 
\item An empirical study of both the frequent star patterns detection and factorization techniques using existing benchmarks. Experimental results show that both E.FSP and G.FSP identify frequent star patterns. Moreover, G.FSP overcomes E.FSP by reducing execution time in at least three orders of magnitude. More importantly, the experiments indicate that factorizing  frequent star patterns by using surrogate keys enable for the creation of compact RDF graphs that reduce size while preserving the information in the original RDF graph.
\end{inparaenum}  

The article is structured as follows: We motivate our research in Section~\ref{sec:motivationExp}, and present an analysis of the state of the art in Section~\ref{sec:relatedwork}. Our approach is defined in Section~\ref{sec:approach}, while Section~\ref{sec:experiment} reports on the results of the experimental study. Finally, we conclude with an outlook on future work in Section~\ref{sec:conclusion}.

\section{Motivating Example}
\label{sec:motivationExp}

\begin{figure*}[t!]
\label{motivationExpGSPA}
\centering
     \vspace{0pt}\subfloat[Subgraphs per 4 Properties]{
      \includegraphics[width=0.35\textwidth]{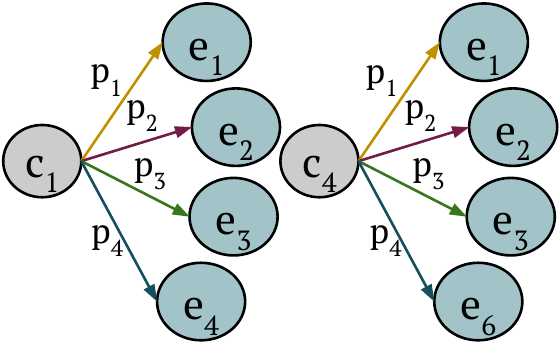}
      \label{fig:gspan4}}
       \vspace{-0.50pt}
       \vspace{0pt}\subfloat[Subgraphs involving three Properties]{
      \includegraphics[width=0.25\textwidth]{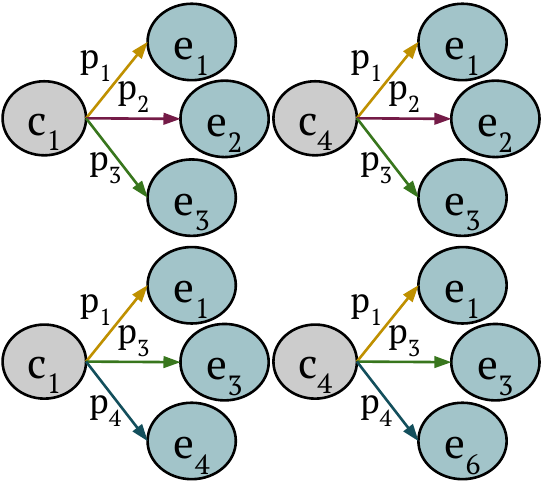}
      \label{fig:gspan3}}
       \vspace{-0.50pt}
  \vspace{0pt}\subfloat[Subgraphs involving two Properties]{
      \includegraphics[width=0.21\textwidth]{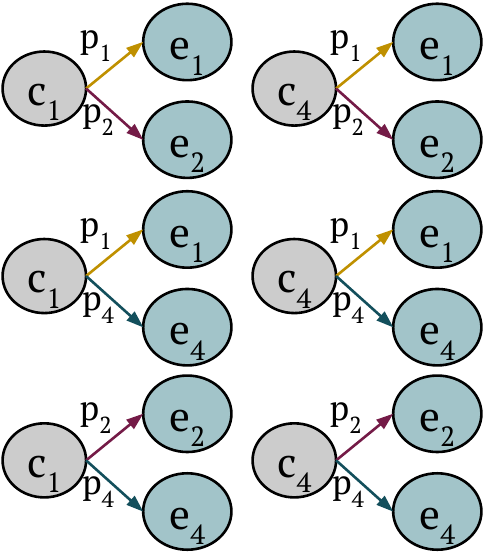}
      \label{fig:gspan2}}
   \caption{{\bf Graph Patterns Identified by gSpan}. Subgraphs, involving entities $c_1$ and $c_4$, extracted by gSpan from the RDF graph in Figure~\ref{fig:duplicates}. (a) Subgraphs per set $\{p_1,p_2,p_3,p_4\}$ of properties; (b) Subgraphs involving three properties from $p_1$, $p_2$, $p_3$, and $p_4$; (c) Subgraphs around two properties from $p_1$, $p_2$, $p_3$, and $p_4$.}
\end{figure*} 

We motivate the problem addressed by this work with an RDF graph where entities of the same type -- or resources -- match the same star pattern. In an RDF graph, matching the same star pattern means that the properties and objects are the same, whereas the entities are different. 
When the number of entities matching a star pattern is very high, the size of the RDF graph increases and the query processing over the RDF graph is affected negatively. A star pattern with a high number of matching entities is a frequent star pattern.
Figure~\ref{fig:duplicates} depicts an RDF graph composed by a class $C$, the entities $c_1$, $c_2$, $c_3$, $c_4$, $e_1$, $e_2$, $e_3$, $e_4$, $e_5$, and $e_6$, and the properties $p_1$, $p_2$, $p_3$, and $p_4$. 
A directed edge $(s\; p\; o)$ in the RDF graph stands for an RDF triple where $p$ is a label that represents an RDF predicate, while $s$ and $o$ are subject and object nodes, respectively.
Edges labeled with the predicate \textit{type}\footnote{property \textit{type} refers to \textit{rdf:type}}, indicate that $c_1$, $c_2$, $c_3$ and $c_4$ are of the same type, i.e., the class $C$. The directed edge $(c_1\; p_1\; e_1)$ expresses that the entity $c_1$ is related to object $e_1$ with the property $p_1$. Similarly, entities $c_2$, $c_3$, and $c_4$ are related to object $e_1$ with the property $p_1$, i.e., the indegree of $e_1$ is four. Similarly, entities $c_1$, $c_2$, $c_3$ and $c_4$ are related to $e_2$ and $e_3$ with the properties $p_2$ and $p_3$, respectively. Note that entities $c_1$, $c_2$, $c_3$, and $c_4$ are associated with the same objects, i.e., $e_1$, $e_2$ and $e_3$ through the edges annotated with same properties $p_1$, $p_2$, and $p_3$. 
Albeit sound, these redundant labeled edges generate frequent star patterns because entities of the same type are described using the same properties and objects. 
Figure~\ref{fig:fgp} illustrates the RDF subgraphs that map to the same star pattern, shown in Figure~\ref{fig:redundancy}, extracted from the RDF graph in Figure~\ref{fig:duplicates}; note that \textit{?x} is a variable whose instantiations correspond to constants in the RDF graph. In these RDF subgraphs, the properties $p_1$, $p_2$, and $p_3$, and the corresponding objects $e_1$, $e_2$, and $e_3$, respectively, are the same, whereas the entities $c_1$, $c_2$, $c_3$, and $c_4$ are different. This indicates that the star pattern is a frequent star pattern, i.e., several entities $c_1$, $c_2$, $c_3$, and $c_4$ instantiate the star pattern. Thus, several entities are related to the same objects, even not all the properties of the class are involved in frequent star patterns. A frequent star pattern comprising the entities $c_1$, $c_2$, $c_3$, and $c_4$ is illustrated in Figure~\ref{fig:redundancy}, where the node $?x$ represents the entities $c_1$, $c_2$, $c_3$, and $c_4$ of class $C$ in the RDF graph in Figure~\ref{fig:duplicates}. 
gSpan~\cite{yan2002gspan} solves the problem of identifying the frequent subgraphs that involve same subject entities related to the same object values using a set of properties. However, our approach requires the identification of frequent star patterns, where each star pattern-- with a subject variable--involves different subject entities related to the same object values using a set of properties. Figures~\ref{fig:gspan4} ,\ref{fig:gspan3}, and \ref{fig:gspan2} show some of the subgraphs extracted by gSpan involving entities $c_1$ and $c_4$ , and the sets of properties containing four, three, and two properties, respectively, from the RDF graphs in Figure~\ref{fig:duplicates}. gSpan exhaustively enumerates the frequent subgraphs; thus, finding frequent star patterns requires an exhaustive search over the generated frequent subgraphs. 
In this work, we exploit the RDF model and propose a technique that allows for transforming an RDF graph $G$ into another RDF graph $G'$ where the number of frequent star patterns is minimized. The graph $G'$ includes all the nodes from $G$ but additionally, $G'$ comprises nodes that represent \textit{factorized entities}-- like the one in Figure~\ref{fig:elemnetsNumber}.

\section{Related Work}
\label{sec:relatedwork}
Database and Semantic Web communities have proposed several representations to speed up processing over the large amounts of data represented using relational and RDF data models~\cite{abadi2006integrating,allen2019understanding,alvarez2011compressed,fernandez2013binary,joshi2013logical,karim2017large,meier2008towards,pichler2010redundancy,zukowski2006super}. 
These compression approaches can be categorized into compression techniques for relational and RDF graph data models. 
Relational data model approaches~\cite{abadi2006integrating,zukowski2006super} efficiently store very large datasets in column-oriented stores. Approaches~\cite{alvarez2011compressed,fernandez2013binary,joshi2013logical,karim2017large,meier2008towards,pan2014graph,pichler2010redundancy,zhu2018predicate} target the efficient storage of RDF graph data. Furthermore, several frequent pattern mining algorithms~\cite{elseidy2014grami,yan2002gspan} extract frequent isomorphic graph patterns from a graph.

%\cite{alvarez2011compressed,fernandez2013binary,joshi2013logical,karim2017large,meier2008towards,pan2014graph,pichler2010redundancy,zhu2018predicate} are proposed to reduce data redundancy in RDF graphs. Some approaches \cite{joshi2013logical,karim2017large,meier2008towards,pichler2010redundancy} generate compact graphical representations of RDF graphs by reducing the number of RDF triples. Other approaches \cite{alvarez2011compressed,fernandez2013binary,pan2014graph,zhu2018predicate} encode RDF resources into compact form representations rather than reducing the number of RDF triples in RDF graphs.

\subsection{Data Compression for Relational Data Models}
Column-oriented databases~\cite{stonebraker2005c,zukowski2006super} store each attribute in a separate column such that successive values of the attribute are accumulated consecutively on the disk. This improves the query processing when the values of some of the columns are required to process the query.  The column oriented data storage opens a number of opportunities to apply compression techniques more naturally over the multiple values of the same type. Compression approach proposed by Abadi et al.~\cite{abadi2006integrating} compress each column in C-store~\cite{stonebraker2005c} using one of the methods like Null Suppression, Dictionary Encoding, Run-length Encoding, Bit-Vector Encoding or Lempel-Ziv~\cite{roth1993database,westmann2000implementation}. Zukowski et al.~\cite{zukowski2006super} focus on improving bad CPU/cache performance caused by the compression techniques involving if-then-else statements in the code, e.g., Null Suppression, Run-length Encoding, and does not take advantage of the super-scalar properties, e.g., pipe-lining the processes, in the modern CPUs. Zukowski et al. propose three compression methods i.e., PFOR, PFOR-DELTA, and PDICT. 
%PFOR expresses values as positive offsets from a base value, PFOR-DELTA represents values as differences from some frame of reference, and PDICT creates a dictionary for the array position of values.
These compression solutions are exploited by column-oriented stores using the decomposition storage model~\cite{copeland1985decomposition}, where \emph{n-array} relations are decomposed into \emph{n} binary relations. Each binary relation consists of one attribute values and the corresponding surrogate keys.  
In this model, two copies of data are stored increasing the data storage requirements. Further, for each attribute a copy of the corresponding duplicated surrogate key is required resulting in an increase of the storage by a factor of two. 
%Moreover, various compression techniques for the same value attributes, are hard to implement with different surrogate keys, except for multi-valued attributes where surrogate keys are repeated.
Moreover, various compression techniques for a large number of unique values, i.e., subject entities, are hard to implement.
Our approach generates a factorized graph where entities matching a frequent star pattern are represented by a surrogate entity of the corresponding compact RDF molecule. These compact graph representations replace repeated properties and corresponding objects with properties and objects in the compact RDF molecules, hence, improve the storage space requirements for the decomposition storage model~\cite{copeland1985decomposition}.

\subsection{Data Compression for the RDF Data Model}
Meier et al.~\cite{meier2008towards} propose a user-specific minimization technique based on Datalog rules to remove the RDF triples from a given RDF graph. Similarly, Pichler et al.~\cite{pichler2010redundancy} study the RDF redundancy elimination in the presence of rules, constraints, and queries specified by users. These two approaches are user specific and require human input for compressing the ever growing RDF graphs. A scalable lossless RDF compression technique, proposed by Joshi et al.~\cite{joshi2013logical}, automatically generates decompression rules. The rules are used to split the RDF datasets into an active dataset containing compressed triples, and a dormant dataset consisting of uncompressed RDF triples. This technique requires the overhead of decompression over the compressed data to access the information initially represented in datasets. A factorized representation of RDF graphs is presented by Karim et al.~\cite{karim2017large}, where repeated observation values are represented only once. This approach reduces the number of RDF triples in the observational data, which is semantically described using the Semantic Sensor Network (SSN) Ontology~\cite{compton2012ssn}.   
We propose an approach to automatically identify frequent star patterns in RDF graphs described using any ontology. Further, we devise factorized graphical representations of RDF graphs which do not require data decompression to perform data management tasks.
Fern\'{a}ndez et al.~\cite{fernandez2013binary} present a binary RDF representation format consisting of a Header, a Dictionary and a Triple component containing RDF metadata, RDF terms catalog, and compactly encoded RDF triples, respectively. Pan et al.~\cite{pan2014graph} propose RDF compression based on graph patterns, which reduces the number of RDF triples and then generates compact binary representations of the reduced triples. The compression technique k\textsuperscript{2}-triples presented by \'Alvarez-Garc\'ia et al.~\cite{alvarez2011compressed} exploits the two dimensional k\textsuperscript{2}-trees structure, proposed by Barisaboa et al.~\cite{brisaboa2009k}, to distribute the compact triples obtained by Header-Dictionary-Triples partitioning~\cite{fernandez2013binary}. These approaches are able to effectively reduce redundancies in RDF graphs, and provide effective techniques for RDF graph compression. However, customized engines are required to perform query processing over the compressed RDF graphs, and decompression techniques are needed during data management.
We devise factorization techniques that use semantics encoded in RDF data and compactly represent RDF triples, reduce redundancy, and facilitate data management tasks without requiring any decompression or a customized engine.

\subsection{Graph Mining Techniques}   
The problem of frequent pattern mining involves finding subgraphs, from a graph, that have frequency above a given threshold. gSpan~\cite{yan2002gspan} exploits the depth first search (DFS) to mine frequent patterns. gSpan maps a graph to a DFS code representing the edges sequence. Several DFS codes can be generated for a single graph. These DFS codes are ordered lexicographically based on the edge labels and the order of nodes being visited. From these ordered DFS codes the minimum DFS codes are selected to build the DFS tree. DFS over a code tree discovers all the minimum DFS codes of frequent patterns. GRAMI~\cite{elseidy2014grami} mines frequent patterns and finds only the minimal set of instances that satisfy the given frequency threshold.  GRAMI stores the templates of frequent patterns instead of storing their appearances. This avoids the creation and storage of all appearances of patterns. For frequency evaluation, GRAMI maps the frequent patterns mining problem to constraint satisfaction problem (CSP), which is represented by a tuple; (a) an ordered set of variables representing nodes, (b) a set of domains of variables in (a), and (c) a set of constrains between these variables. Two subgraphs patterns are isomorphic if the variables in corresponding CSP tuple have different values from the domains, however, nodes and edge labels are the same. Notwithstanding these frequent pattern mining approaches are able to identify the frequent isomorphic graph patterns, extracting frequent star patterns, which involve different subject nodes related with same objects nodes using same set of edge labels, requires an exhaustive search over the identified frequent patterns. It is important to highlight that although these approaches effectively mine subgraph patterns, they are not able to identify patterns where one node is a variable. Contrary, our approach searches for star patterns and is able to detect the ones with highest instantiations.

\section{RDF Graph Factorization Approach}
\label{sec:approach}

We introduce important preliminary definitions, and then formally define the problem of detecting frequent star patterns and compacting them in an RDF graph. 

\subsection{Preliminaries}
Our approach is based on the RDF data model building on RDF triples.
\begin{definition}[RDF triple~\cite{arenas2009foundations} ] 
Let $\mathbf{I}$, $\mathbf{B}$, $\mathbf{L}$ be disjoint infinite sets of URIs, blank nodes, and literals, respectively.
A tuple $(s\; p\; o) \in (\mathbf{I} \cup \mathbf{B}) \times \mathbf{I} \times (\mathbf{I} \cup \mathbf{B} \cup \mathbf{L})$ is an RDF triple, where $s$ is the subject, $p$ is the property, and $o$ is the object. %
\end{definition} 

A set of RDF triples is called RDF dataset (or knowledge graph) and can also be viewed as a graph. Thus, in Figure \ref{fig:duplicates}, the edge $(c_1\; type\; C)$ represents an \textit{RDF triple}, where entity $c_1$ corresponds to subject, $type$ and $C$ represent a property and an object, respectively; there are nineteen more \textit{RDF triples}.

\begin{definition}[RDF Graph]
An RDF graph $G=(V,E,L)$ is a labeled directed graph where nodes represent entities or objects, while labels stand for properties:
\begin{itemize}
    \item An RDF triple $(s\; p\; o) \in E$, corresponds to an edge in $E$ from node $s$ to node $o$; $p$ is the label of the edge and denote the property that relates both nodes;
    \item $s$, $o$ $\in V$, $s$ corresponds to a subject and $o$ corresponds to an object; and
    \item $p \in L$, is an edge label corresponding to a property.
\end{itemize}
\end{definition}

%The RDF graph, in Figure~\ref{fig:duplicates}, represents links from entities $c_1$, $c_2$, $c_3$ and $c_4$, of class $C$ defined using property $type$, to the objects $e_1$, $e_2$, and $e_3$, using directed labeled edges annotated with the properties $p_1$, $p_2$, and $p_3$, respectively, and to the objects $e_4$, $e_5$ and $e_6$, using directed labeled edges annotated with property $p_4$. The entities $c_1$, $c_2$, $c_3$ and $c_4$ stand for $subject$, and nodes $e_1$, $e_2$, $e_3$, $C$, $e_4$, $e_5$, and $e_6$ denote $object$ in the corresponding directed labeled edges in the RDF graph. We now introduce the notion of RDF molecules, which comprise sets of triples sharing the same subject.

\begin{definition}[RDF Molecule~\cite{FernandezLC14}]
An RDF molecule $RM$ is a set of RDF triples that share the same subject, i.e., $RM$= $(s \; p_1 \; o_1)$,$(s \; p_2 \; o_2)$,$\dots$,$(s \; p_n \; o_n)$.
\end{definition}

Figure~\ref{fig:fgp} presents four RDF molecules around the subjects $c_1$, $c_2$, $c_3$, and $c_4$ of class $C$. In the RDF molecule around subject $c_1$ all the RDF triples describe $c_1$ using properties $p_1$, $p_2$, and $p_3$. Similarly, RDF triples in each of the other RDF molecules describe the subjects $c_2$, $c_3$, and $c_4$ using properties $p_1$, $p_2$, and $p_3$.

%to node $o$ in the RDF graph node $s$ corresponds to a ,  the node $o$ corresponds to an \textit{object} and the edge label $p$ corresponds to a \textit{property}. 

\subsection{Problem Statement}
Star patterns denote graph patterns covering RDF molecules:
\begin{definition}[Star Pattern] 
Given is an RDF graph $G=(V,E,L)$, a class $C$ in $E$ and a set of properties $SP= \{p_1, p_2,\dots, p_n\}$ such that $C$ is the domain of all the properties in $SP$. 
Let entities $o_1,o_2,\dots,o_n$ be the objects of the properties $p_1, p_2,\dots, p_n$, respectively. Let $?s$ be a variable. 
A star pattern of $C$ over the properties $p_1, p_2,\dots, p_n$ and objects $o_1,o_2,\dots,o_n$ corresponds to a graph pattern composed of the conjunction of triple patterns: $(?s \; p_1 \; o_1)$,$(?s \; p_2 \; o_2),\dots,(?s \; p_n \; o_n)$. 

\end{definition}

Figure~\ref{fig:redundancy} shows a star pattern composed of three triple patterns containing properties $p_1$, $p_2$, and $p_3$ and the corresponding objects $e_1$, $e_2$, and $e_3$, respectively. The entities $c_1$, $c_2$, $c_3$, and $c_4$ of class $C$ in the RDF graph in Figure~\ref{fig:duplicates} match the star pattern. The  variable $?x$ is the subject of the triple patterns referring to the entities matching the star pattern.

\begin{definition}[Class Multiplicity] 
Given an RDF graph $G=(V,E,L)$, a class $C$ in $E$ and a set of properties $SP= \{p_1, p_2,\dots, p_n\}$ such that $C$ is the domain of all the properties in set $SP$ of properties.
Let entities $o_1,o_2,\dots,o_n$ be objects of the properties $p_1, p_2,\dots, p_n$, respectively. 
The multiplicity of $o_1$, $o_2$,$\dots$, $o_n$ in $G$,  $M(o_1,o_2,\dots,o_n|G)$ is defined as the number of different entities in $C$ that match a star pattern having the same objects $o_1,o_2,\dots,o_n$ in the properties $p_1,p_2,\dots,p_n$. Entities $s$ correspond to instantiations of the subject variable in the star pattern.  
\[\arraycolsep=1.2pt
\begin{array}{ll}
M(o_1, o_2,\dots,o_n|G)=  |\{s|&(s \;\texttt{:type} \; C) \in G, \\&(s\; p_1 \; o_1) \in G, (s\; p_2 \; o_2) \in G,
\\&\dots,(s\; p_n \; o_n) \in G \}|
\end{array}
\]
\end{definition}

In the RDF graph in Figure \ref{fig:duplicates}, the multiplicity of the objects $e_1$, $e_2$ and $e_3$, given the set $\{p_1,p_2,p_3\}$ of properties, is $4$, because there are four instantiations of the subject variable. Similarly, the multiplicity of objects $e_4$, $e_5$ and $e_6$, in the set $\{p_4\}$ of properties is $1$ and $2$.

\begin{definition}[Class Multiplicity Inverse] 
Given class $C$, a set $SP=\{p_1,p_2,\dots,$ $p_n\}$ of properties and corresponding objects $o_1$, $o_2$,$\dots$, $o_n$, the multiplicity inverse of $o_1$, $o_2$,$\dots$, $o_n$ 
in $G$,  denoted $MI(o_1,o_2,\dots,o_n|G)$, is:
\[MI(o_1,o_2,\dots,o_n|G)=1/M(o_1, o_2,\dots,o_n|G)\]
\end{definition}

In the RDF graph in Figure \ref{fig:duplicates}, the \textit{class multiplicity inverse} of the objects $e_1$, $e_2$, and $e_3$, given the set $\{p_1,p_2,p_3\}$ of properties, is $\frac{1}{4}$. The multiplicity inverse of objects $e_4$, $e_5$, and $e_6$ in the set $\{p_4\}$ of properties is $\frac{1}{1}$ and $\frac{1}{2}$.

\begin{definition}[Multiplicity of Star Patterns] 
Given a class $C$ in an RDF graph $G$ with properties $SP=\{p_1,p_2\dots, p_n\}$. The multiplicity of the star patterns in $C$ over $SP$, $AMI_G(p_1,p_2,\dots,p_n|C)$, is defined as follows:
\[ \]
\[\arraycolsep=1pt
\begin{array}{ll}
AMI_G(p_1,p_2,\dots,p_n|C) = & \lceil  f'_{\forall s \in C} (\{MI(o_1, o_2,\dots,o_n|G)\\&|(s \;\texttt{type} \; C) \in G, (s\; p_1 \; o_1) \in G,\\&  (s\; p_2 \; o_2) \in G,\dots,(s\; p_n \; o_n) \\&\in G \})\rceil
\end{array}
\]
where $f'(.)$ is an aggregation (e.g., summation) function.
\end{definition}

In the RDF graph in Figure~\ref{fig:duplicates}, the \textit{multiplicity of the star patterns} of $C$ over the set $\{p_1,p_2,p_3\}$ of properties is $\frac{1}{4}+\frac{1}{4}+\frac{1}{4}+\frac{1}{4}=1$, which is obtained by summing up the class multiplicity inverse of the objects $e_1$, $e_2$, and $e_3$ given the set $\{p_1,p_2,p_3\}$ of properties, for each entity $c_1$, $c_2$, $c_3$, and $c_4$ of class $C$ matching the star pattern. Similarly, the \textit{multiplicity of the star patterns} of class $C$ over the set $\{p_4\}$ is $\frac{1}{2}+\frac{1}{2}+\frac{1}{1}+\frac{1}{1}=3$ in the RDF graph, and is obtained by summing up the individual class multiplicity inverse of objects $e_4$, $e_5$, and $e_6$ given the set $\{p_4\}$, for each of the entities $c_1$, $c_2$, $c_3$, and $c_4$ of class $C$ that map the corresponding star patterns. The \textit{multiplicity of the star patterns} over a set of properties corresponds to the number of star patterns composed of the set of properties and the corresponding objects. The problem of frequent star patterns detection is defined next, the solutions correspond to frequent star patterns. We define the frequent star patterns detection problem as the minimization of connections between a class instances and values linked through the properties. To find the minimum number of edges over the properties in a class, the sum of the number of edges in the star patterns over a set of properties and the number of edges between the class entities and the properties that are not involved in the star patterns is computed.  
\begin{definition}[FSP Detection Problem]
Given an RDF graph $G=(V,E,L)$ and a class $C$ in $G$ with set of properties $S$ and number of instances $AM_G(C)$. The problem of \textit{Frequent Star Patterns Detection (FSP Detection)} 
is to find a subset $SP$ of $S$ such that the star patterns $SGP$ of $C$ over $SP$ corresponds to \emph{frequent star patterns}, i.e., $\#Edges(SP,C,G)$ is minimized:

 \begin{equation}
\label{eq:minProblem}
\tiny
 \argmin_{\mathit{SP}\subseteq \mathit{S}}{\{\underbrace{AMI_G(SP|C)*(|SP|+1)+AM_G(C)*(|S-SP|)}_{\#Edges(SP,C,G)}\}}
 \end{equation}
 \end{definition}

   \begin{figure*}[t!]
\centering
     \subfloat[$\#Edges(SP,C,G)$ over $p_1$,$p_2$,$p_3$, and $p_4$]{
      \includegraphics[width=.25\textwidth]{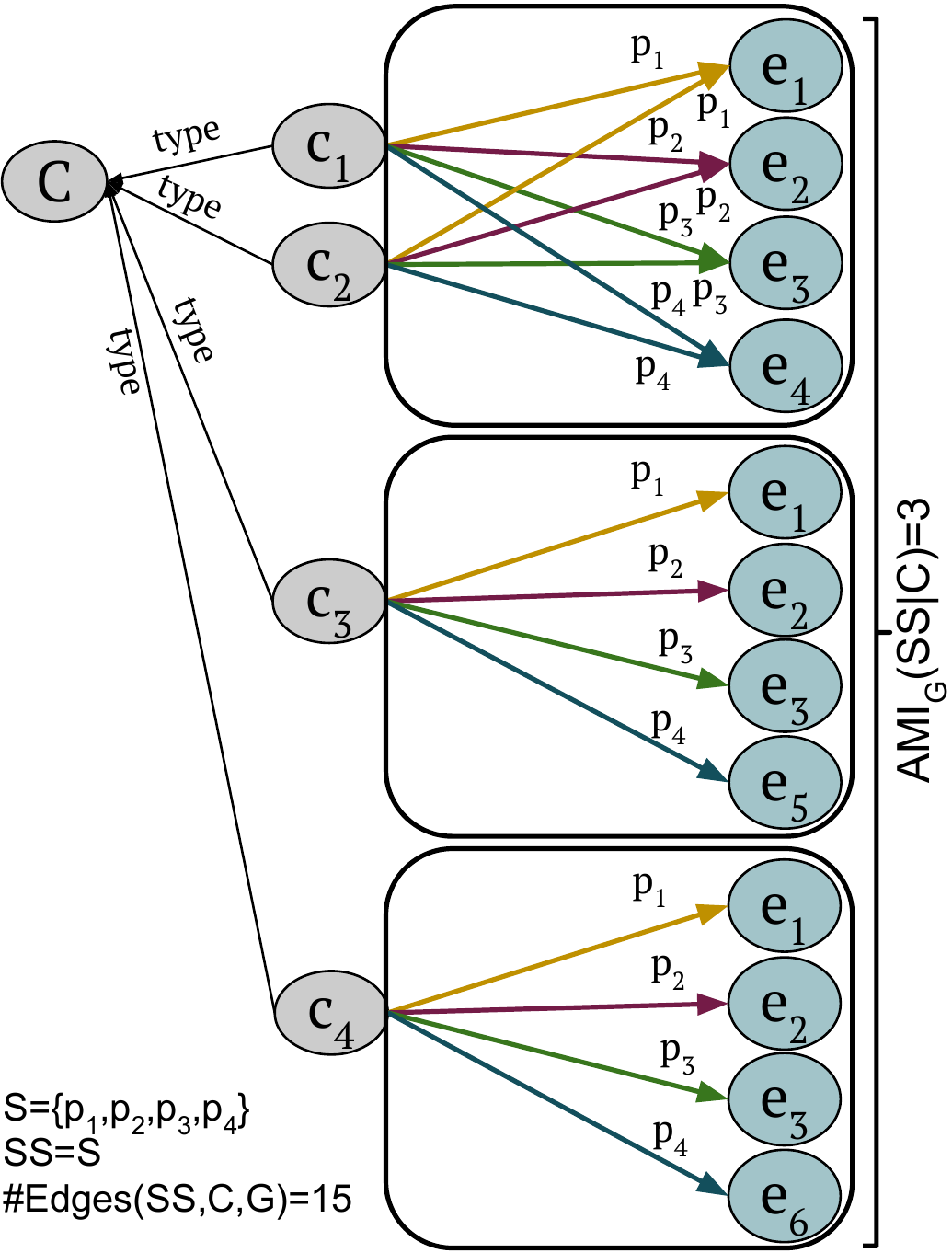}
      \label{fig:fspP1p2p3p4}}\hspace*{0.2em}
      \subfloat[$\#Edges(SP,C,G)$ over properties $p_1$,$p_2$, and $p_3$]{
      \includegraphics[width=.3\textwidth]{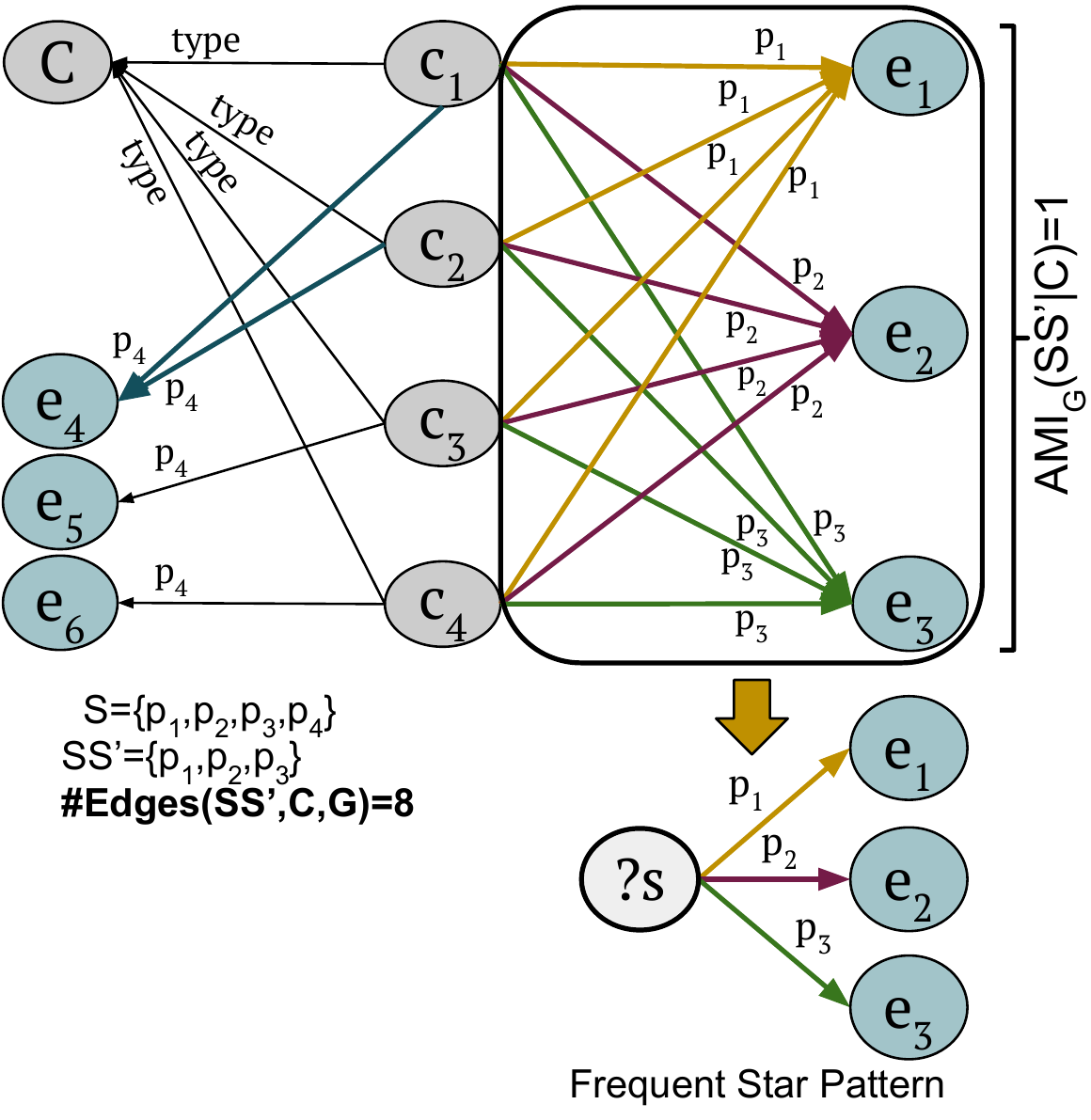}
      \label{fig:fspP1p2p3}}\hspace*{0.2em}
      \vspace{0pt}\subfloat[{\bf Factorized Graph } $G'$ from $G$]{
      \includegraphics[width=.36\textwidth]{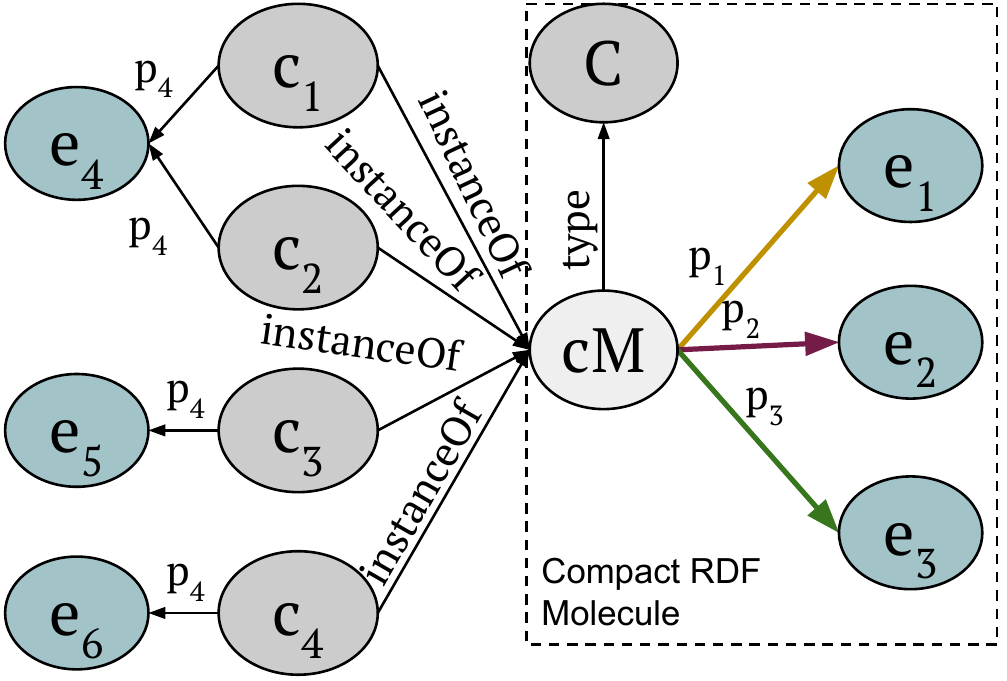}
      \label{fig:noduplicates}}
    \caption{{\bf The Frequent Star Patterns Detection Problem}. Properties involved in frequent star patterns. (a) Stars patterns over the set $SS=\{p_1,p_2,p_3,p_4\}$ of properties in class $C$ require three surrogate entities and $\#Edges(SS,C,G)$ are 15; (b) Star patterns over the set $SS'=\{p_1,p_2,p_3\}$ of properties in class $C$ require one surrogate entitiy and $\#Edges(SS',C,G)$ are eight; (c)  A factorized RDF graph $G'$ of $G$ composed of compact RDF molecule with a surrogate entity $cM$.}
    \label{fig:fspProblemExmp}
\end{figure*}

 % \begin{figure*}[t!]
%\centering
    %
 %    \vspace{0pt}\subfloat[Minimum edge sets of labeled edges over properties $p_1$, $p_2$, and $p_3$, with a high cardinality five]{
  %    \includegraphics[width=.52\textwidth]{edgeGroupsp1p2p3p4.pdf}
  %    \label{fig:edgeGroups}}\hspace*{0.3em}
  %    \vspace{0pt}\subfloat[Entities matching the Star Pattern]{
  %    \includegraphics[width=.25\textwidth]{entityGroupp1p2p3.pdf}
  %    \label{fig:entitySet}}\hspace*{0.3em}
  %    \vspace{0pt}\subfloat[A Compact RDF Molecule]{
  %    \includegraphics[width=.16\textwidth]{molecule.pdf}
  %    \label{fig:elemnetsNumber}}
  %  \caption{{\bf The Frequent Star Patterns Detection Problem}. Properties in frequent star patterns. (a) Minimum edge sets, over $p_1$, $p_2$, and $p_3$, with a high cardinality shown in lower right corner; (b) Entities map the star pattern over $p_1$, $p_2$, and $p_3$; (c) A compact RDF molecule for the star pattern.}
  %  \label{fig:sets}
%\end{figure*}     

Figure~\ref{fig:fspProblemExmp} illustrates the problem of detecting frequent star patterns from the RDF graph in Figure~\ref{fig:duplicates}. Figure~\ref{fig:fspP1p2p3p4} presents three star patterns $AMI_G(SS|C)$ over the set of properties $p_1$, $p_2$, $p_3$, and $p_4$, and 15 edges in $\#Edges(SS,C,G)$.  However, only one star pattern $AMI_G(SS'|C)$ over the set of properties $p_1$, $p_2$, and $p_3$ exists in Figure~\ref{fig:fspP1p2p3}. A small value of $\text{\small{\#Edges(SS',C,G)}}$ i.e., eight, shows a subgraph over $SS'$ that is represented by only one star pattern with more instantiations than the star patterns for $SS$, i.e., it is a frequent star pattern. Thus, the set of properties $SP$ where  $\text{\small{\#Edges(SP,C,G)}}$ is minimal, encloses a subgraph with the minimal number of star patterns which have the maximal number of instantiations; additionally, these star patterns are the ones with the greater number of properties.  Figure~\ref{fig:noduplicates} depicts the factorized RDF graph where this frequent star pattern has been replaced with a compact RDF molecule on a surrogate entity $cM$; this factorization reduces the size of the original RDF graph.

\begin{theorem}
\label{theorem:formula}
Given an RDF  graph $G$, a class $C$ in $G$, and non-empty sets of properties $S$, $SP$, and $SP'$ of $C$ such that $SP'\subset SP \subset S$. If $\#Edges(SP',C,G) > \#Edges(SP,C,G)$, then $\forall SP'' \subset SP'$, $\#Edges(SP'',C,G) \geq \#Edges(SP,C,G)$.
\end{theorem}
\begin{proof}
By contradiction. Suppose $\#Edges(SP'',C,G) < \#Edges(SP,C,G)$. From $\#Edges(SP',C,G) > \#Edges(SP,C,G)$ and $SP'\subset SP \subset S$, it can be inferred that $AMI_G(SP|C)<AM_G(C)$, $AMI_G(SP'|C)<AM_G(C)$, $|SP''|<|SP'|<|SP|<|S|$, $|SP-SP''|\geq2$, and $AMI_G(SP''|C)<AM_G(C)$. Considering these inequalities in  $\#Edges(SP'',C,G)$ and $\#Edges(SP,C,G)$, we can demonstrate that $\#Edges(SP'',C,G)$ is at least greater than $\#Edges(SP,C,G)$ in $2*AM_G(C)$, contradicting, thus, $\#Edges(SP'',C,G) < \#Edges(SP,C,G)$.
\end{proof}

 \begin{figure*}[t!]
\centering
     \vspace{-0.1cm}\subfloat[{\bf \textmu\textsubscript{N}} from $G$ into $G'$]{
      \includegraphics[width=.28\textwidth]{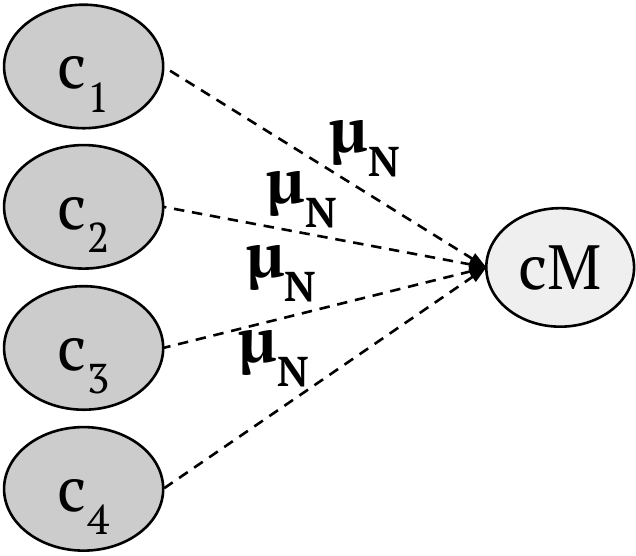}
      \label{fig:mapNode}}
      \vspace{-0.1cm}\subfloat[{\bf $type$} $G$ into $instanceOf$ $G'$]{
      \includegraphics[width=.32\textwidth]{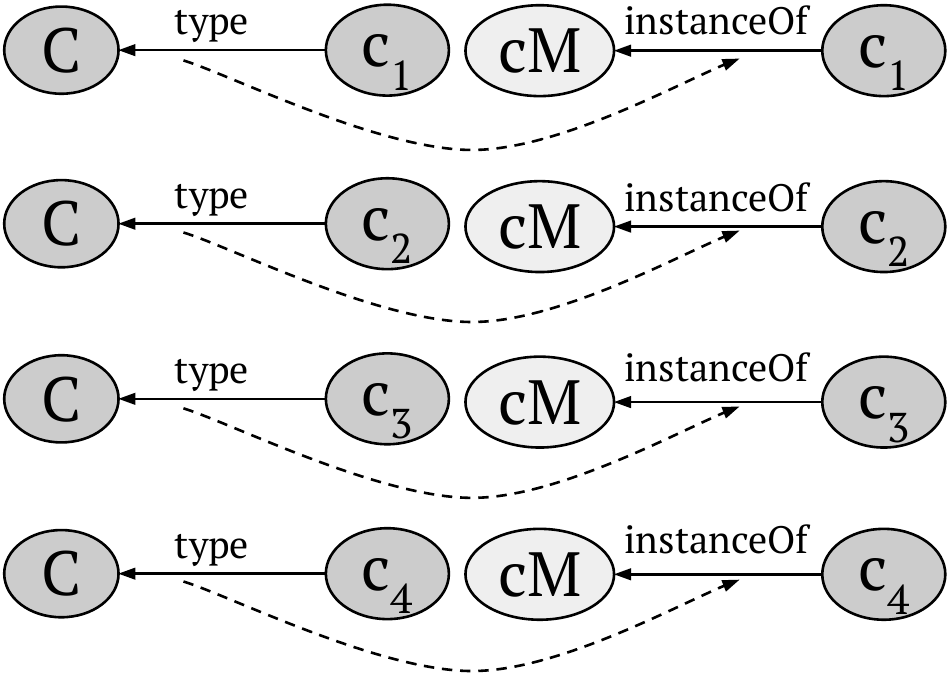}
      \label{fig:mapEdge}}
        \vspace{-0.1cm}\subfloat[A Compact RDF Molecule]{
      \includegraphics[width=.2\textwidth]{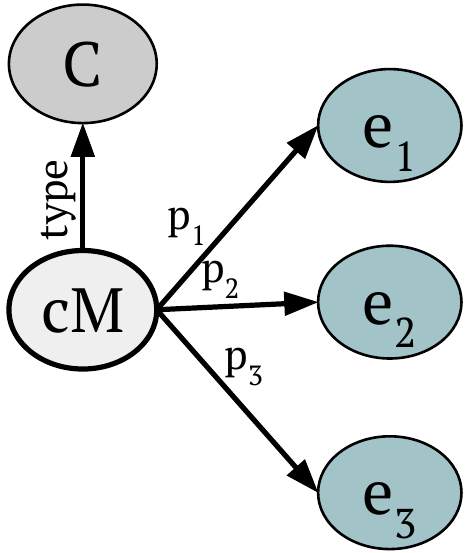}
      \label{fig:elemnetsNumber}}
    \caption{{\bf The RDF Graph Factorization Problem}. Factorization of RDF graph $G$ into $G'$. (a) Entity mappings $\mu_N$ from the RDF graph $G$ in \ref{fig:duplicates}  to the surrogate entity $cM$ in $G'$; (b) Transformation of property $type$ from $G$ to $G'$; (c) A compact RDF molecule for the frequent star pattern over the properties $p_1$, $p_2$, and $p_3$.}
    \label{fig:FRDFProblem}
\end{figure*} 
 
\begin{definition}[A Compact RDF Molecule]
Given a star pattern \textit{SGP} of a class $C$ 
over the properties $p_1, p_2,\dots, p_n$ and objects $o_1,o_2,\dots,o_n$. Given a surrogate entity $sg$ of type $C$. A compact RDF molecule for \textit{SGP} is an RDF molecule composed of RDF triples $(sg \; p_1 \; o_1)$,$(sg \; p_2 \; o_2)$,$\dots$,$(sg \; p_n \; o_n)$.
\end{definition}

Figure~\ref{fig:elemnetsNumber} shows a compact RDF molecule that instantiates the star pattern presented in Figure~\ref{fig:redundancy}, which is composed of the properties $p_1$, $p_2$, and $p_3$ and the corresponding objects $e_1$, $e_2$, and $e_3$, respectively. The surrogate entity $cM$ in the compact RDF molecule, represents the entities $c_1$, $c_2$, $c_3$, and $c_4$ of type $C$ matching the star pattern, as shown in Figure~\ref{fig:fgp}.

\begin{definition}[The RDF-F Problem]
Given an RDF graph $G=(V,E,L)$ and a set of properties $SP$, the problem of \textit{RDF factorization (RDF-F)} corresponds to finding \textit{a factorized RDF graph} of $G$, $G'=(V',E',L')$, where the following hold:
\begin{itemize}
\item Entities in $G$ are preserved in $G'$, i.e., $V \subseteq V'$. 
\item For each entity $s_i$ in $V$ that corresponds to an instantiation of the variable of a frequent star pattern $SGP$ of a class $C$ over the set $SP$ in $G$, there is an entity $s_{SGP}$ in $V'$ that corresponds to the surrogate entity of the compact RDF molecule of $SGP$. Formally, there is a partial mapping $\mu_N$: $V \rightarrow V'$:
\begin{itemize}
\item Instances of the frequent star pattern $SGP$ are mapped to the surrogate entity of the star pattern, i.e., $\mu_N(s_i)$=$s_{SGP}$.
    \item The mapping  $\mu_N$ is not defined for the rest of the entities that do not instantiate a frequent star pattern in $G$.
\end{itemize}
\item For each RDF triple $t$ in ($s \; p \; o$) in $E$:
\mbox{}
\begin{itemize}
\item If $\mu_N(s)$ is defined and $C_s$ is the type of $s$, and $p$ is $type$, then the triples  ($s \; \textit{instanceOf} \; \mu_N(s)$), ($\mu_N(s)\; type \; C_s$) belong to $E'$.
\item If $\mu_N(s)$ is defined and $C_s$ is the type of $s$, and $p \in SP$, then the triples  ($\mu_N(s)\; p \; o $) belong to $E'$.
\item Otherwise, the RDF triple $t$ is preserved in $E'$.
\end{itemize}
\end{itemize}
\end{definition}

Consider RDF graphs $G$ and $G'$ shown in Figures~\ref{fig:duplicates} and \ref{fig:noduplicates}, respectively.
Figure~\ref{fig:mapNode} depicts a map $\mu_N$ that assigns entities $c_1$, $c_2$, $c_3$, and $c_4$ of class $C$ in $G$ to the surrogate entity $cM$ in $G'$. Further, entities $c_1$, $c_2$, $c_3$, $c_4$, $C$, $e_1$, $e_2$, $e_3$, $e_4$, $e_5$, and $e_6$ are preserved in $G'$. Moreover, the edge labeled with property $p_1$ in $G$, i.e., ($c_1\; p_1 \; e_1$) is presented with edges ($c_1\; instanceOf \; cM$), ($cM\; p_1 \; e_1$) and ($cM\; type \; C$) in $G'$; similarly, edges labeled with properties $p_2$ and $p_3$ in $G$ are represented in $G'$. Figure~\ref{fig:mapEdge} shows the transformations of the connections between entities $c_1$, $c_2$, $c_3$, and $c_4$ and the class $C$ using labeled edges annotated with property $type$, with the connections relating the entities $c_1$, $c_2$, $c_3$, and $c_4$ to the corresponding surrogate entity $cM$ using the property $instanceOf$.

\begin{definition}[Axioms for InstanceOf]
The property \emph{instanceOf} is a functional property defined as follows:
\begin{itemize}
\item If ($s_i \; \textit{instanceOf} \; sg $) and 
($sg \; \textit{type} \; C $) then
($s_i \; \textit{type} \; C$).
\item If ($s_i \; \textit{instanceOf} \; sg $) and 
($sg \; p_j \; o_k $) then
($s_i \;  p_j \; o_k $).
\end{itemize}
\end{definition}

These two axioms enable to represent implicitly, all the knowledge encoded in the edges from an original RDF graph that are removed during the factorization process. They are utilized during query processing to rewrite queries over the original RDF graph into queries against the factorized RDF graph.

\subsection{FSP Detection Approach}

\begin{algorithm}
\caption{E.FSP Algorithm}
\label{algo:exhfsgDetection}
%\vspace*{-.5cm}\begin{multicols}{2}
\begin{algorithmic}[1]
  \REQUIRE A dictionary $\textit{subgraphsDict}$ of subgraphs over the subsets of properties in $S$, A set $S$ of properties of class $C$.
  \ENSURE Frequent star patterns $\textit{fsp}$, A set $SP$ of properties
  \STATE $\textit{fsp} \leftarrow []$, $SP \leftarrow \emptyset$, $\textit{minEdges} \leftarrow 0$, $\textit{subsetCard} \leftarrow |S|$
  \WHILE{$\textit{subsetCard} \geq 2$}
    \STATE $\textit{propSets} \leftarrow \textit{getSubsetsOf}(S, \textit{subsetCard})$
    \FOR{$SP \in \textit{propSets}$}
        \STATE  $\textit{subgraphs} \leftarrow \textit{subgraphsDict}[SP]$
        \STATE $\textit{totalEdges} \leftarrow \textit{countEdges}(\textit{subgraphs})$
        \IF{$\textit{minEdges} == 0$}
            \STATE $\textit{minEdges} \leftarrow \textit{totalEdges}$
            \STATE $\textit{fsp} \leftarrow \textit{subgraphs}$
            \STATE $\textit{bestSP} \leftarrow SP$
        \ELSIF{$\textit{totalEdges} < \textit{minEdges}$}
            \STATE $\textit{minEdges} \leftarrow \textit{totalEdges}$
            \STATE $\textit{fsp} \leftarrow \textit{subgraphs}$
            \STATE $\textit{bestSP} \leftarrow SP$
        \ENDIF
    \ENDFOR
    \STATE $\textit{subsetCard} \leftarrow \textit{subsetCard} - 1$
  \ENDWHILE
  \STATE $SP \leftarrow \textit{bestSP}$
\RETURN{$\textit{fsp}$, $SP$ }
\end{algorithmic}
%\end{multicols}
%\vspace*{-.4cm}
\end{algorithm}

\begin{figure*}[t!]
\centering
      \subfloat[Exhaustive FSP Approach (E.FSP)]{
      \includegraphics[width=.5\textwidth]{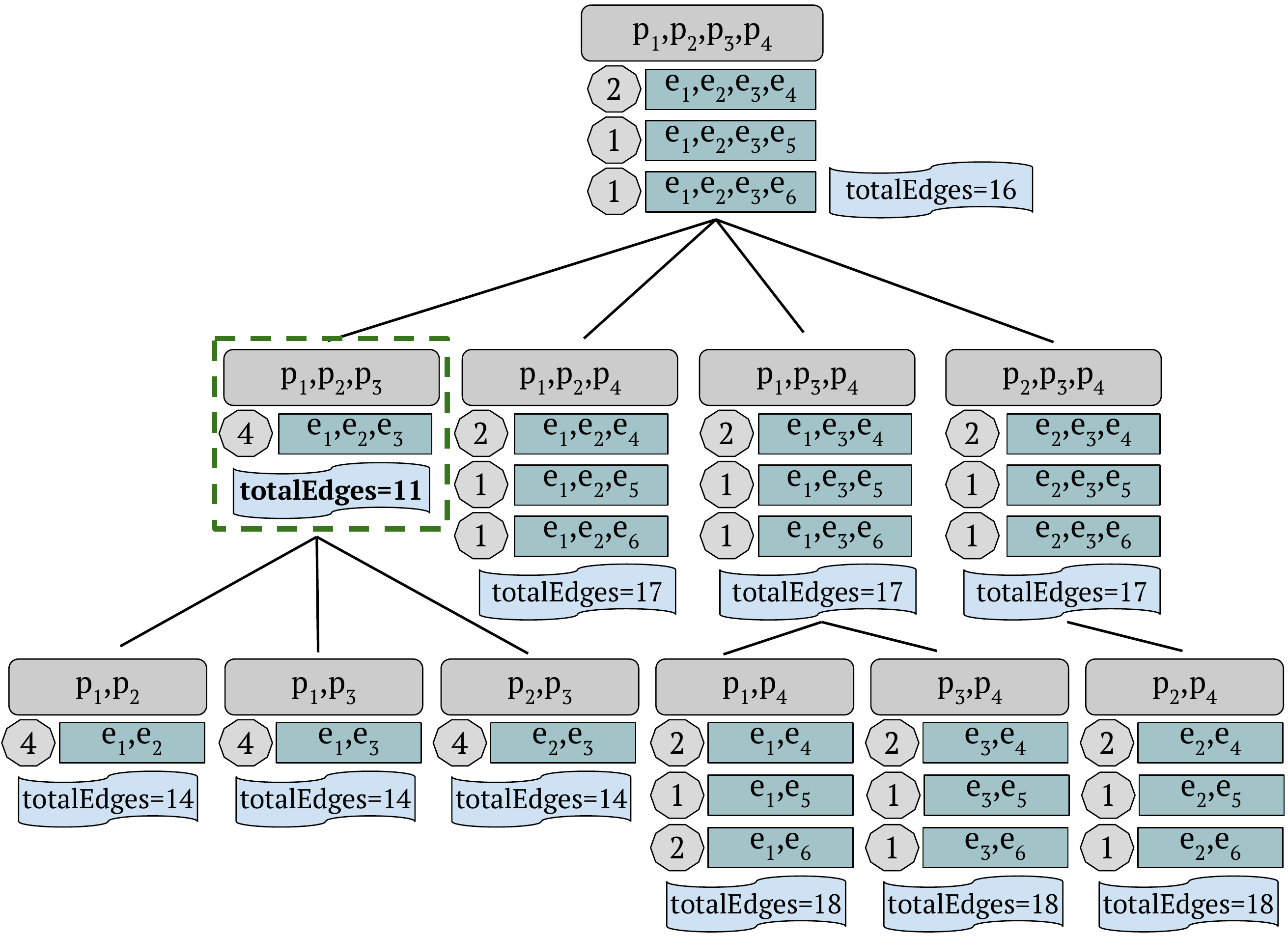}
      \label{fig:exhExmp}}
      \subfloat[Greedy FSP Approach (G.FSP)]{
      \includegraphics[width=.5\textwidth]{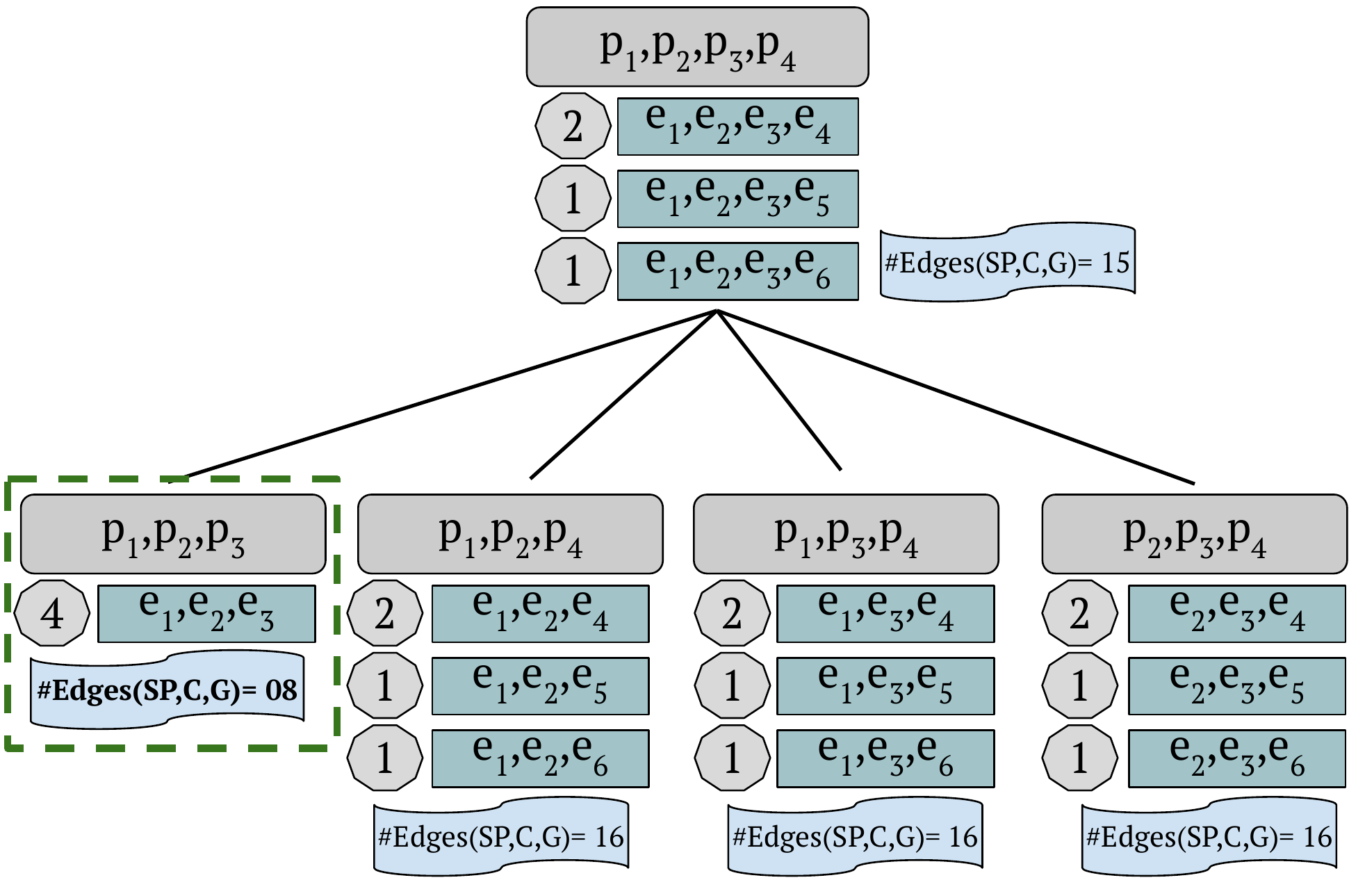}
      \label{fig:greedyExmp}}
\caption{{\bf Frequent Star Patterns Detection}. \emph{E.FSP} and \emph{G.FSP} iterate over the star patterns in the RDF graph in Figure~\ref{fig:duplicates} to detect the frequent star patterns. (a) \emph{E.FSP} exhaustively iterates over the whole search space of frequent patterns; (c) \emph{G.FSP} iterates the search space without generating all the star patterns.}
    \label{fig:algoExamples}
\end{figure*}

To solve the \emph{FSP detection} problem, we propose two algorithms that perform iterations over frequent patterns involving different sets of properties sets of a class $C$ in an RDF graph $G$, and the class entities.
\emph{E.FSP}, presented in Algorithm~\ref{algo:exhfsgDetection}, resorts to a frequent pattern mining algorithm like gSpan. \emph{E.FSP} exploits breadth first search technique to exhaustively traverse the search space of frequent patterns generated by the frequent pattern mining algorithm, and always finds the best frequent star patterns. 
Figure~\ref{fig:exhExmp} illustrates the iterations performed by \emph{E.FSP} to find the frequent star patterns in the RDF graph in Figure~\ref{fig:duplicates}. \emph{E.FSP} receives a dictionary $\textit{subgraphsDict}$ of all the subgraphs over the subsets of the set $S$ of properties in the class $C$ in an RDF graph $G$. The keys of the dictionary $\textit{subgraphsDict}$ are the combination of properties in the subsets of $S$, and the dictionary values are the subgraphs involving the properties from the corresponding keys.
\emph{E.FSP} generates frequent star patterns and a set of properties involved in the frequent star patterns. 
\emph{E.FSP} initializes the variables $\textit{fsp}$, $SP$, $\textit{minEdges}$, and $\textit{subsetCard}$ in line 1. The variables $\textit{minEdges}$ and $\textit{subsetCard}$ are initialized with values $0$ and cardinality of $S$, respectively. From lines 2-18, \emph{E.FSP} iterates over all the subgraphs involving two or more properties to find the frequent star patterns. In Figure~\ref{fig:exhExmp}, \emph{E.FSP} starts iterations with the set of properties $SP=\{p_1,p_2,p_3,p_4\}$, and the subgraphs involving the properties in subsets of $SP$, where the cardinality of subsets is equal to the cardinality of $S$, i.e., four (line 3). The generated subset contains all the properties in $SP$, i.e., $\{p_1,p_2,p_3,p_4\}$, and the total number of edges \textit{totalEdges} in $SP$ is computed i.e., 16 (line 5-6). Since $\textit{minEdges}$ are $0$, therefore, the value 16 of $\textit{totalEdges}$ is assigned to $\textit{minEdges}$, subgraphs over $SP=\{p_1,p_2,p_3,p_4\}$ and $SP$ are assigned to $fsp$ and $bestSP$, respectively (line 7-10). 
At line 17, the subset size $subsetSize$ is reduced by one in order to generate the subsets of properties of $S$ with the cardinality one less the cardinality of $S$ i.e., three.  The subsets $\{p_1,p_2,p_4\}$, $\{p_1,p_3,p_4\}$, and $\{p_2,p_3,p_4\}$, of cardinality three, generate more number of edges i.e., value of $\textit{totalEdges}$ is 17, than the minimum number of edges $\textit{minEdges}$, i.e., 16, and are not selected as the best sets of properties. However, the subgraphs over the subset $\{p_1,p_2,p_3\}$ contain 11 number of triples, which is less than 16 the value of $\textit{minEdges}$. Therefore, \emph{E.FSP} selects $\{p_1,p_2,p_3\}$ as the best set of properties and the corresponding subgraphs as the frequent star patterns (line 11-15). Once all the subsets $SP$ of $S$ with cardinality three, are evaluated, the value of $subsetCard$ is reduced by one i.e., two, and the subsets of cardinality two are evaluated in the next iteration. Figure~\ref{fig:exhExmp} presents that all the subsets of cardinality two generate larger values, i.e., 14 and 18, for $\textit{totalEdges}$ than the value 11 for $\textit{minEdges}$. Therefore, none of the subsets of properties of cardinality two contains the frequent star patterns. Further, all the subsets of cardinality greater or equal to two have been evaluated, \emph{E.FSP} stops and returns $\{p_1,p_2,p_3\}$ as the best set of properties and the corresponding subgraphs as the frequent star patterns (line 19-20).

\begin{algorithm}
\caption{G.FSP Algorithm}
\label{algo:greedyfsgDetection}
%\vspace*{-.5cm}\begin{multicols}{2}
\begin{algorithmic}[1]
    \REQUIRE A set $S$ of properties of class $C$ in $G$, and a list $\textit{starList}$ of star patterns over properties in $S$. 
  \ENSURE Frequent star patterns \textit{fsp}, A set $SP$ of properties.
  \STATE $\textit{fsp} \leftarrow []$, $\textit{starList}' \leftarrow []$, $SP \leftarrow S$, $SP' \leftarrow \emptyset$, $\textit{fValue} \leftarrow \textit{fValue}' \leftarrow 0$
  \REPEAT 
    \IF{$|SP|\geq2$}
        \IF{$AMI_G(SP|C)==1$}
            \STATE $\textit{fsp} \leftarrow \textit{starList}$   
            \RETURN {\textit{fsp}, $SP$}
        \ELSE
            \STATE $\textit{fValue} \leftarrow \#\textit{Edges}(SP,C,G)$
            \FOR {$p \in SP$}
                \STATE   $SP' \leftarrow SP-\{p\}$
                \IF{$|SP'|\geq2$}
                    \STATE   Create $\textit{starList}'$ over $SP'$ using $\textit{starList}$
                    \STATE  $value \leftarrow \#Edges(SP',C,G)$
                    \IF{$AMI_G(SP'|C) ==1$}
                        \STATE  $\textit{fValue}' \leftarrow \textit{value}$
                        \STATE $\textit{bestSP} \leftarrow SP'$
                        \STATE $\textit{bestSList} \leftarrow \textit{starList}'$
                        \STATE \textit{break}
                    \ELSIF{$\textit{value} < \textit{fValue}'$}
                        \STATE $\textit{fValue}' \leftarrow \textit{value}$
                        \STATE $\textit{bestSP} \leftarrow SP'$
                        \STATE $\textit{bestSList} \leftarrow \textit{starList}'$
                    \ENDIF
                \ENDIF
            \ENDFOR
        \ENDIF
    \ENDIF
        \STATE $\textit{starList} \leftarrow \textit{bestSList}$, $SP \leftarrow \textit{bestSP}$
 \UNTIL{$\textit{fValue}' > \textit{fValue}$}
 \STATE $\textit{fsp} \leftarrow \textit{starList}$
\RETURN{$\textit{fsp}$, $SP$}
\end{algorithmic}
%\end{multicols}
%\vspace*{-.7cm}
\end{algorithm}

\emph{G.FSP}, presented in Algorithm~\ref{algo:greedyfsgDetection}, adopts a greedy algorithm to traverse the search space without generating all the frequent patterns. \emph{G.FSP} starts iterations using a set $SP$ of properties containing all the properties in $S$ of a class $C$ in an RDF graph $G$. \emph{G.FSP} computes the value of Formula~\ref{eq:minProblem} for $SP$ and iterates over the subsets $SP'$ of cardinality one less the cardinality of $SP$ and computes Formula~\ref{eq:minProblem} for each of subsets $SP'$.
A property subset $SP'$ with a smaller formula value than the formula value of $SP$, is selected as the best set of properties in that iteration, and is used in the next iteration to check the subsets of cardinality one less the cardinality of the selected set of properties. The iterations are performed until the cardinality of the selected subset of properties is less than two. 
Based on the property presented in Theorem~\ref{theorem:formula}, \emph{G.FSP} stops, if none of the subsets $SP'$ generates less value for formula than the formula value of $SP$.
In addition, \emph{G.FSP} stops whenever the cardinality of the set of properties is less than two, or the multiplicity of star patterns $AMI_G(SP|C)$ is one.
\emph{G.FSP} receives a set $S$ of properties in class $C$ in an RDF graph $G$, and a list $starList$ of star patterns involving properties in $S$.
\emph{G.FSP} returns frequent star patterns $\textit{fsp}$ and a set of properties $SP$ involved in the frequent star patterns.
Figure~\ref{fig:greedyExmp} shows the iterations performed by \emph{G.FSP} to detect the frequent star patterns in the RDF graph in Figure~\ref{fig:duplicates}. \emph{G.FSP} initializes all the variables at line 1, where $SP$ is assigned the set $S$ of properties for the first iteration i.e., $SP=\{p_1,p_2,p_3,p_4\}$. In lines 2-29, \emph{G.FSP} iterates over the subsets of $SP$ to find the frequent star patterns based on the criteria in Formula~\ref{eq:minProblem}. The cardinality value four of $SP$ is greater than two (line 3), and $AMI_G(SP|C)$ is not equal to one (line 4-7), therefore, \emph{G.FSP} computes the value of $\#Edges(SP,C,G)$ of $SP$ i.e., $15$ (line 8). In lines 9-25, \emph{G.FSP} iterates over the subsets of $SP$ of cardinality one less the cardinality of $SP$ to find the best set of properties for the next iteration. At line 10, a property $p$ is removed from $SP$ to generate a subset $SP'$ e.g., by removing $p_1$ a subset $SP'=\{p_2,p_3,p_4\}$ is generated. Since the cardinality of $SP'$ is more than two, therefore, a star list $starList'$, representing the star patterns over $SP'$, is created using $starList$ (line 12-13). The value of $\#Edges(SP',C,G)$ for $SP'$ is computed i.e., 16 (line 13). For $SP'$, $AMI_G(SP'|C)$ is not one, and the value 16 of $\#Edges(SP',C,G)$ for $SP'$ is not less than the value 15 of $\#Edges(SP,C,G)$ for $SP$, therefore, the star patterns over $SP'=\{p_2,p_3,p_4\}$ do not involve frequent star patterns and $SP'$ is not a best candidate for the next iteration. Similarly, the property subsets $\{p_1,p_3,p_4\}$ and $\{p_1,p_2,p_4\}$, generated from $SP$ by removing $p_2$ and $p_3$, respectively,  give a higher value 16 for $\#Edges(SP',C,G)$ and the star patterns over these set of properties are not better than the star patterns over $SP$. However, $SP'=\{p_1,p_2,p_3\}$, generated from $SP$ by removing $p_4$, gives one star pattern, therefore, the star pattern involving properties in $SP'$ is returned as the frequent star pattern without performing more iteration (line 14-18). 
In case, the set $SP'$ of properties is involved in more than star patterns and the formula value of $SP'$ smaller than the value of $SP$, then $SP'$ is selected for the next iteration (line 19-23).
\emph{G.FSP} stops and no further iterations are performed if none of the subsets $SP'$ of $SP$ generates a smaller value for $\#Edges(SP',C,G)$ than $\#Edges(SP,C,G)$. \emph{G.FSP} returns the star patterns involving $SP$, with a minimum value for $\#Edges(SP,C,G)$, as the frequent star patterns, and $SP$ as the best set of properties.
\emph{E.FSP} and \emph{G.FSP} work under the following assumptions: (a) all RDF molecules are complete, i.e., all class entities have values for all the properties, (b) all the properties are functional. In addition to these assumptions, \emph{G.FSP} has one more assumption: (c) if there are ties while deciding between the sets of  properties, only one will be selected.
Complexity of \emph{E.FSP} is exponential, i.e., $2^n$. \emph{G.FSP} adopts a Greedy approach and prunes the search space by selecting only the best set of properties during each iteration until the stop condition is met, i.e., no better set of properties with a minimum formula value can be found. In the worst case, the computational complexity of \emph{G.FSP} is $\sum_{i=0}^{n} (n-i)$=$\frac{n(n+1)}{2}$, where $n$ is the cardinality of the input set of properties. The complexity of \emph{G.FSP} grows linearly with the increase in the size of the input set of properties.

\begin{figure*}[t!]
\centering
     \vspace{0pt}\subfloat[ Transformation Rules for Class $C$]{
      \includegraphics[width=.5\textwidth]{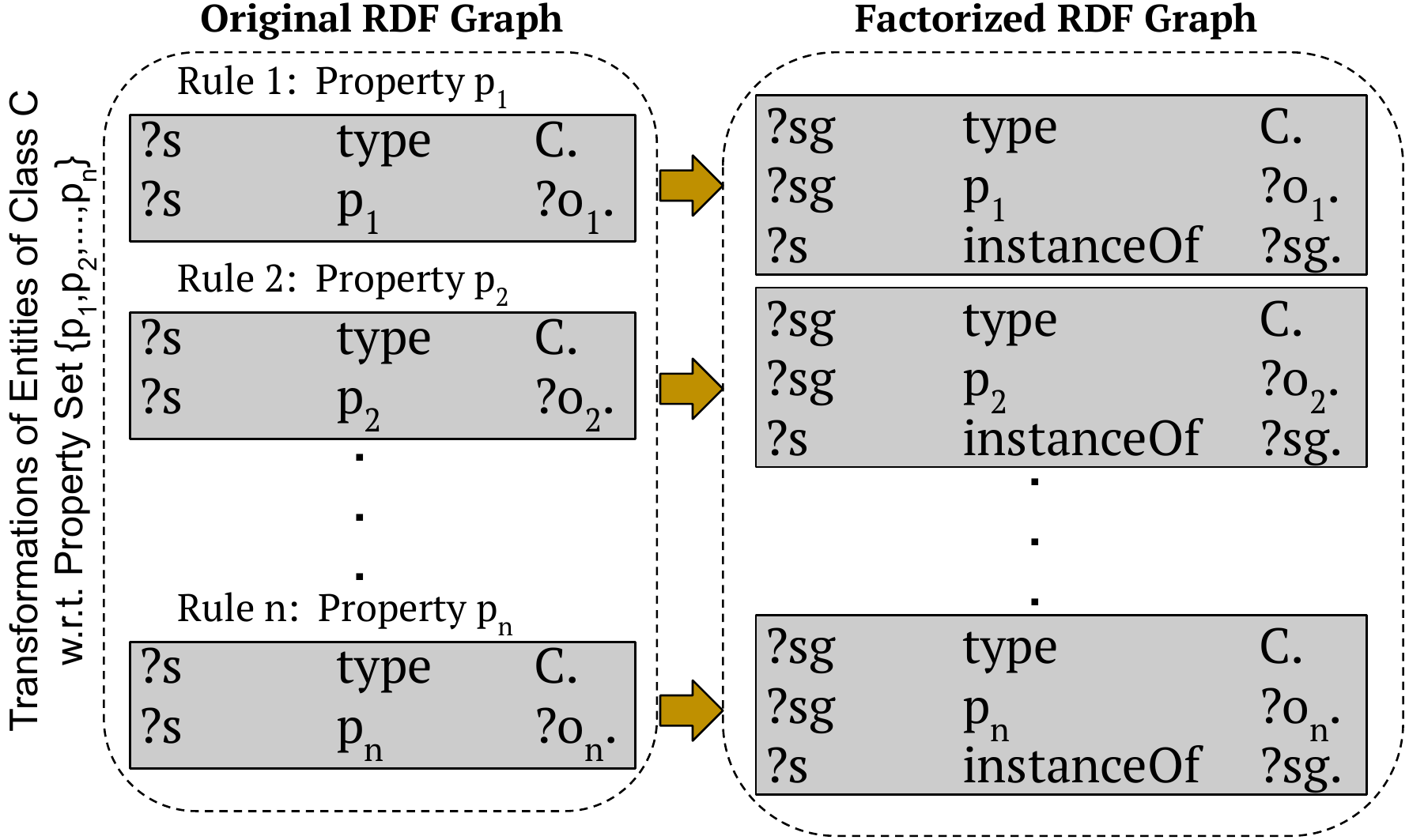}
      \label{fig:rules}}
      \vspace{0pt}\subfloat[Original and Factorized RDF Graphs]{
      \includegraphics[width=.49\textwidth]{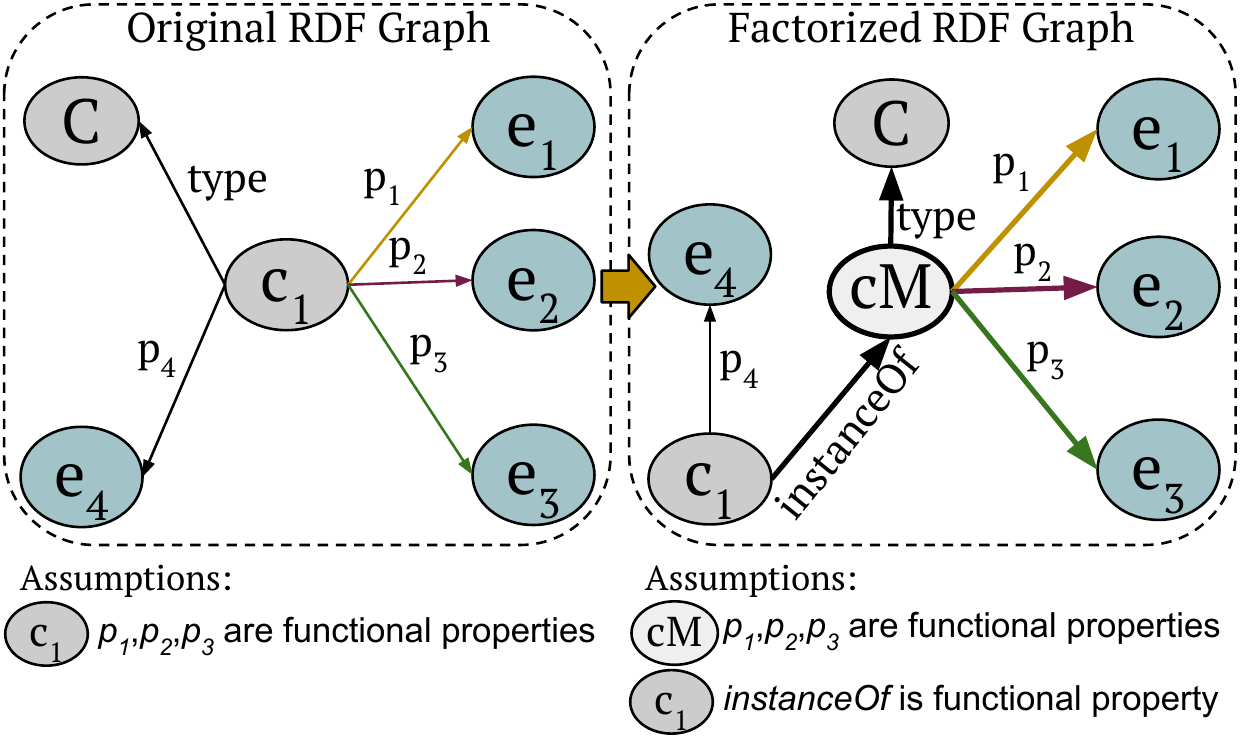}
      \label{fig:graphtransform}}
    \caption{{\bf Transformations in RDF Graph}. Transformation rules preserved between original and factorized RDF graphs. (a) Transformation rules over the properties $p_1,p_2,\dots,p_n$; (b) Portions of RDF graphs (original and factorized). Nodes and edges added to create the factorized RDF graph, are highlighted in bold.}
    \label{fig:transformation}
\end{figure*}

\subsection{A Factorization Approach}

\begin{algorithm}
\caption{The Factorization Algorithm}
\label{algo:factorization}
%\vspace*{-.5cm}\begin{multicols}{2}
\begin{algorithmic}[1]
  \REQUIRE An RDF graph $G(V,E,L)$, A class $C$, A set $SP$ of properties from \emph{E.FSP} Algorithm~\ref{algo:exhfsgDetection} or \emph{G.FSP} Algorithm~\ref{algo:greedyfsgDetection}
  \ENSURE  Factorized\;RDF\;Graph $G'(V',E',L')$ and entity mappings $\mu_N$
  \STATE $\mu_N \longleftarrow \emptyset, V' \longleftarrow \emptyset, E' \longleftarrow \emptyset, L' \longleftarrow \emptyset$
  
  \FORALL{$o_1, o_2,\dots,o_n \in V such\; that\; SS=\{s|p_1,p_2,\dots,p_n \in SP AND$ \\ 
$(s\; \texttt{type}\; \texttt{C}) \in G, (s\; {p_1}\; o_1) \in G,
(s\; {p_2}\; o_2) \in G \dots, (s\; {p_n}\; o_n) \in G\}$}
\STATE $sg \leftarrow SurrogateEntity()$
\FOR{$ss \in SS$}
\STATE $\mu_N\leftarrow\mu_N \cup \{(ss,sg)\}$ 
\ENDFOR
\ENDFOR
\FOR{$(s\; p\; o) \in E \land s, o \in V$}
    \IF{$\mu_N(s) \neq \emptyset$} 
        \STATE \COMMENT {Create compact RDF molecule}
        \IF{$p == type$}
            \STATE $E' \leftarrow E' \cup \{(s\;instanceOf\; \mu_N(s)),$ $(\mu_N(s)\;p\;o)\}$
             \STATE $V' \leftarrow V' \cup \{s,\mu_N(s),o\}$
            \STATE $L' \leftarrow L' \cup \{p,instanceOf\}$
        \ELSIF{$p \in SP$}
            \STATE $E' \leftarrow E' \cup \{(\mu_N(s)\;p\;o)\}$
             \STATE $V' \leftarrow V' \cup \{\mu_N(s),o\}$
            \STATE $L' \leftarrow L' \cup \{p\}$
        \ELSE
            \STATE $E' \leftarrow E' \cup \{(s\;p\;o)\}$
            \STATE $V' \leftarrow V' \cup \{s,o\}$
            \STATE $L' \leftarrow L' \cup \{p\}$
        \ENDIF
    \ELSE 
        \STATE $E' \leftarrow E' \cup \{(s\;p\;o)\}$
        \STATE $V' \leftarrow V' \cup \{s,o\}$
        \STATE $L' \leftarrow L' \cup \{p\}$
    \ENDIF
\ENDFOR
\RETURN {$G'(V',E',L'),\mu_N$ }
\end{algorithmic}
%\end{multicols}
%\vspace*{-.4cm}
\end{algorithm}

We present a solution to the problem of factorizing RDF graphs describing data using ontologies.
A sketch of the proposed method is presented in Algorithm~\ref{algo:factorization}.
The algorithm receives an RDF graph $G=(V,E,L)$, a class $C$, and a set $SP'$ of properties from \emph{E.FSP} or \emph{G.FSP}, and generates a factorized RDF graph $G'=(V',E',L')$, and the entity mappings $\mu_N$ from the entities of class $C$ in $V$ in RDF graph $G$ to the surrogate entities in $V'$ in RDF graph $G'$. The algorithm initializes the set of mappings $\mu_N$, the set of nodes $V'$, the set of labeled edges $E'$ and the set of edge labels (properties) $L'$ of the factorized RDF graph $G'$ (line 1).
For all the entities of $C$ related to the same objects $o_1,o_2,\dots,o_n$ using edges annotated with properties $p_1,p_2,\dots,p_n$ in $SP'$, the algorithm creates a surrogate entity $sg$ for the corresponding compact RDF molecule in $G'$ (lines 2-3). 
In lines 4-6, the algorithm maps all the entities, that are related to $o_1,o_2,\dots,o_n$ using properties $p_1,p_2,\dots,p_n$ in $G$, to the surrogate entity in $\mu_N$.
Once all the mappings of the entities of $C$ in $G$ to the corresponding surrogate entities in $G'$ are in $\mu_N$, the factorized RDF graph $G'$ is created using $\mu_N$ (lines 8-29).
For each RDF triple $(s\;p\;o)$ in $E$, if entity mapping $\mu_N(s)$ is defined, then a compact RDF molecule is created. If $p$ is $type$, then the new edges $(s\; instance\;\mu_N(s))$, and $(\mu_N(s)\; p\; o)$ are added to $G'$ along with entities $s$, $o$ and the mapped surrogate entity of $s$, and the edge labels $p$ and $instanceOf$ (lines 11-14). If $p$ is in $SP$, the new edge $(\mu_N(s)\; p\; o)$, and entities $s$, $o$ and the edge label $p$ are added to $G'$ (lines 15-18). 
If entity mapping $\mu_N(s)$ are defined, however, $p$ is not in $SP$, or $p$ is not $type$, then the edge $(s\; p\; o)$ is added to $G'$ along with the corresponding nodes and the edge label (lines 19-23).  
If entity mapping $\mu_N(s)$ is not defined, then the edge $(s\; p\; o)$ and the nodes $s$ and $o$, and edge label $p$ are added to $G'$ (lines 24-28).

\begin{figure*}[t!]
\centering
     \vspace{0pt}\subfloat[\%age Decrease in Edges after Factorization]{
      \includegraphics[width=.5\textwidth]{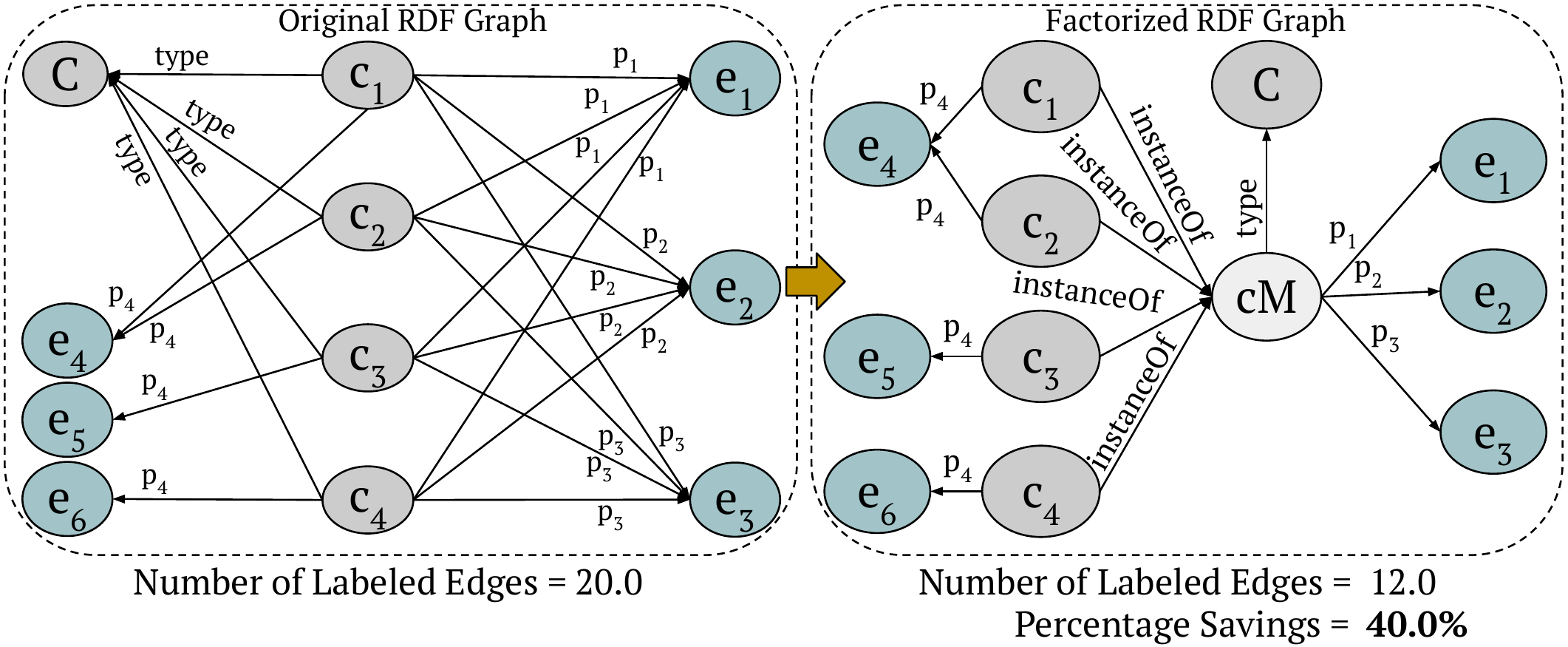}
      \label{fig:decrease}}
      \vspace{0pt}\subfloat[\%age Increase in Edges after Factorization]{
      \includegraphics[width=.5\textwidth]{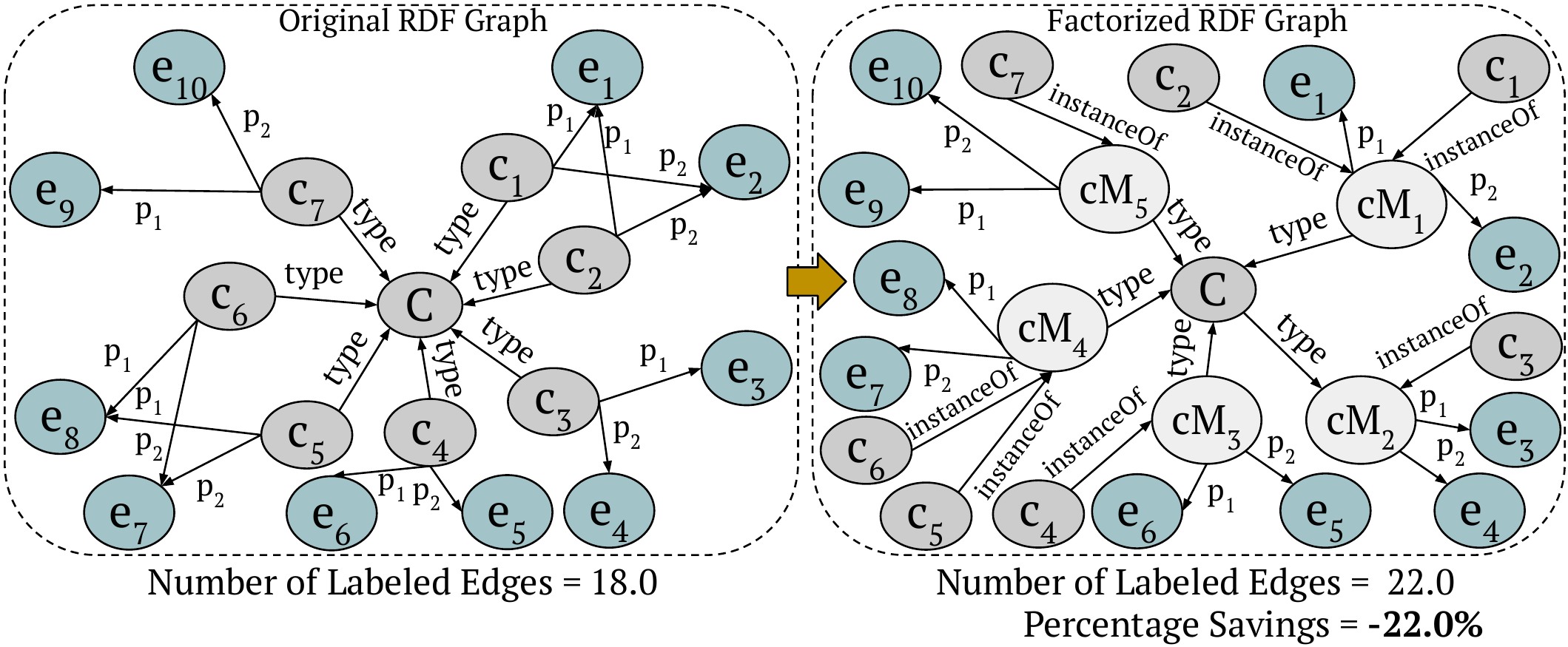}
      \label{fig:increase}}
    \caption{{\bf RDF Graph Factorization Overhead}. Factorization of RDF graphs is not worthy in all cases. (a) Entities of class $C$ in the original RDF graph match the frequent star pattern over the properties $p_1$, $p_2$, and $p_3$; (b) few entities match each star pattern over $p_1$ and $p_2$ causing factorization overhead.}
    \label{fig:overhead}
\end{figure*}

Figure~\ref{fig:transformation} depicts the transformations for the set $\{p_1,p_2,\\\dots,p_n\}$ of properties performed in an RDF graph whenever a corresponding factorized RDF graph is created. Figure~\ref{fig:rules} presents transformation rules; one rule for each property in $\{p_1,p_2,\dots,p_n\}$ of class $C$. Each rule states how the labeled edges associated with a $C$ in an original RDF graph are transformed into the edges in the factorized graph. Rule 1 asserts that the relation between an entity $s$ of $C$ with an object $o_1$ is not explicitly represented by one property in the factorized RDF graph. In order to retrieve $o_1$, a path across the labeled edges between entities $s$ and the corresponding surrogate entity $sg$ have to be traversed. Similarly, the rest of the transformation rules establish how explicit associations between entities of $C$ and the objects using properties $p_2,\dots,p_n$ in the original RDF graphs are represented by the path of labeled edges annotated with properties in the factorized RDF graphs. Algorithm \ref{algo:factorization} adds the corresponding labeled edges of these paths in lines 7-16. Furthermore, Figure~\ref{fig:graphtransform} presents a portion of the RDF graph in Figure~\ref{fig:duplicates} and corresponding transformation in the factorized RDF graph in Figure~\ref{fig:noduplicates}. The surrogate entity and the new labeled edges are highlighted in bold in the factorized RDF graph. 
The Algorithm~\ref{algo:factorization} creates the surrogate entity in line 4; new edges are created in line 10. Additionally, assumptions about the characteristics of the entity associations in the graph are presented. Some edges existing between the entities in RDF graph in Figure~\ref{fig:duplicates} are not present in the factorized RDF graph in Figure~\ref{fig:noduplicates}, these entity associations can be obtained by traversing the factorized RDF graph as indicated by the corresponding transformation rules in Figure~\ref{fig:rules}.

Figure~\ref{fig:overhead} illustrates an example of factorization overhead, i.e., a case when it is not worthy to factorize a class given a set of properties in an RDF graph.
%The sum of the number of nodes and labeled edges in an RDF graph corresponds to the size of the RDF graph. 
Figure~\ref{fig:decrease} presents an example where savings are observed in the number of edges after factorization. The factorization of RDF graph in Figure~\ref{fig:decrease} for the class $C$ using the properties $p_1$, $p_2$, and $p_3$, reduces the number of edges from $20.0$ to $12.0$ and the positive value $40.0\%$ for percentage savings indicates a percentage decrease in the number of edges. Furthermore, the edge savings gained after factorization are high enough to compensate the addition of the surrogate entity $cM$ in the factorized RDF graph.  
In contrast, factorization of the RDF graph over the class $C$ using the properties $p_1$ and $p_2$ introduces an overhead, as shown in Figure~\ref{fig:increase},  by increasing the number of nodes and edges in the factorized RDF graph. The number of edges is increased from $18.0$ to $22.0$, shown in Figure~\ref{fig:increase}, after factorization and a negative value $-22.0\%$ for the percentage savings indicates an increase in the number of edges. The star patterns, detected in the original RDF graph, in Figure~\ref{fig:increase}, are replaced by the corresponding compact RDF molecules with the corresponding surrogate entity and new labeled edges (presented in Algorithm~\ref{algo:factorization}). Due to the high number of star patterns, the addition of the surrogate entities and new labeled edges increases the size of the factorized RDF graph. 
%Moreover, the addition of the surrogate entities $cM_1$, $cM_2$, $cM_3$, $cM_4$, and $cM_5$ is not compensated since no savings are observed in the number of labeled edges in the factorized RDF graph. 

\begin{table*}[b]
\centering
\caption{{\bf Datasets}. (a) Statistics of the datasets with observations about several weather phenomena, collected from around 20,000 weather stations in the United States; (b) The number of labeled edges \textit{NLE(G)}, in the datasets obtained after gradually integrating the RDF datasets D1, D2, and D3 describing observations.}
\vspace{0pt}\subfloat[Statistics of datasets collected from around 20,000 weather stations in the US.]{
	\begin{tabular}{|c| c | c | c | c |}
				\hline
				\textbf{Dataset} &\textbf{Climate} & \textbf{Date} & \textbf{\#RDF} &\textbf{\# Obs} \\ 
				\textbf{ID} &\textbf{Event} & \textbf{} & \textbf{Triples} &\textbf{}   \\ \hline
				\textbf{D1} & Blizzard & April, 2003 & ~38,054,493 & ~~4,092,492 \\
				\textbf{D2} & Hurricane Charley & August, 2004 & 108,644,568 &  11,648,607 \\
				\textbf{D3} & Hurricane Katrina & August, 2005 & 179,128,407 & 19,233,458  \\ \hline
			\end{tabular}
\label{tbl:data}}
\\
\vspace{0pt}\subfloat[Number of Labeled Edges \textit{NLE(G)} in datasets.]{\small
\begin{tabular}{|ccc|}
\hline
\multicolumn{1}{|c}{\bf{Dataset}} & \multicolumn{1}{|c}{\bf{Observation}}&\multicolumn{1}{|c|}{\bf{Measurement}} \\
\multicolumn{1}{|c}{\bf{ID}} &\multicolumn{1}{|c}{\bf{NLE($G$)}}&\multicolumn{1}{|c|}{\bf{NLE($G$)}} \\ \hline
\multicolumn{1}{|l}{\bf{D1}}   & \multicolumn{1}{|r}{24,142,314} & \multicolumn{1}{|r|}{12,071,157}\\ \hline 
\multicolumn{1}{|l}{\bf{D1D2}}  &  \multicolumn{1}{|r}{93,286,824} & \multicolumn{1}{|r|}{46,643,412} \\ \hline
\multicolumn{1}{|l}{\bf{D1D2D3}} & \multicolumn{1}{|r}{207,630,306} & \multicolumn{1}{|r|}{103,815,153} \\\hline
\end{tabular}
\label{tab:NELOriginal}}
\label{tbl:datatsets}
\end{table*}

\begin{table*}[b]
\centering
\caption{{\bf Observation and Measurement Classes}. Sets of properties containing different properties of the \textit{Observation} and \textit{Measurement} (Meas.) classes in the SSN ontology, each set of properties is assigned a unique ID, e.g., \textit{A1} and \textit{A8}.}
\subfloat{
	%\begin{tabular}{|p{15mm}|p{60.4mm}|p{4mm}|}
	\begin{tabular}{|c|p{60.4mm}|p{4mm}|}
\hline
 \multicolumn{1}{|l|}{\bf{Class}} & \multicolumn{1}{l|}{\bf{Set of Properties}}  & \multicolumn{1}{l|}{\bf{SID}} \\
\hline
%\multirow{7}{*}{Observation} & \{property\} & A1 \\ \cline{2-3}
\multirow{7}{*}{\rotatebox[origin=c]{90}{Observation}} & \{property\} & A1 \\ \cline{2-3}
 &  \{time\} & A2\\ \cline{2-3}
 & \{procedure, generatedBy\} & A3  \\ \cline{2-3}
 & \{property, procedure, generatedBy, time\} &  A4 \\ \cline{2-3}
 & \{property, procedure, generatedBy\} & A5 \\ \cline{2-3}
 & \{property, time\} & A6 \\ \cline{2-3}
 & \{procedure, time, generatedBy\} & A7 \\ \hline
\multirow{3}{*}{\rotatebox[origin=c]{90}{Meas.}} & \{value, unit\} & A8 \\ \cline{2-3}
  & \{value\} & A9 \\ \cline{2-3}
 &    \{unit\} & A10 \\ \hline
\end{tabular}
\label{tab:obsattributes}}
\label{tab:attributes}
\end{table*}

\section{Experimental Study}
\label{sec:experiment}
We study the effectiveness of the proposed techniques for detecting frequent star patterns.
%We empirically study the effect of combining properties over the factorization techniques proposed  for the semantically represented data using RDF. 
Moreover, given a class, we evaluate the impact of the factorization techniques over the RDF graphs size by selecting several combinations of the properties in the class. 
%on the size of the factorized RDF graphs and the impact of the proposed factorization techniques over the RDF graph size. 
We empirically assessed the following research questions:
\begin{inparaenum}[\bf {\bf RQ}1\upshape)]
    \item Are the proposed frequent star patterns detection techniques able to efficiently detect the frequent star patterns in RDF graphs?
    \item Are the proposed frequent star patterns detection techniques able to detect the frequent star patterns in RDF graphs?
    \item \\What is the impact of different combinations of properties of a class over the size of factorized RDF graphs?
     %\item What is the impact of different combinations of properties of a class over the RDF graph factorization?
     %\item Are the recommendations of the frequent graph patterns detection approach useful for the factorization of RDF graphs?
     \item Are the proposed factorization techniques able to reduce the number of labeled edges in RDF graphs?
    %\item Is the factorization time affected by the size of the RDF graphs?
\end{inparaenum}
Our experimental configuration is as follows:
~\\
\noindent
\textbf{ Datasets.}  Evaluation is conducted on three \emph{LinkedSensorData} datasets~\cite{patni2010linked} semantically described using the Semantic Sensor Network (SSN) Ontology. 
These RDF datasets comprise observations and measurements of several climate phenomena, e.g., temperature, visibility, precipitation, rainfall, and humidity, collected during the hurricane and blizzard seasons in the United States in the years 2003, 2004, and 2005\footnote{Available at: \url{http://wiki.knoesis.org/index.php/LinkedSensorData}}.  
Table~\ref{tbl:data} describes the main characteristics of these RDF datasets.
Moreover, Figure~\ref{fig:distribution} shows percentage of repeated RDF triples with wind speed, temperature, and relative humidity values in dataset $D1D2D3$. The unit of measurement is same for each type of observation. These plots show that some of the large number of observed values are highly repeated in the datasets. Further, values are not discretized to produce the same query answers.

\begin{figure*}[t!]
\centering
     \vspace{-0.5pt}\subfloat[\%age of Windspeed Repeated Triples in $D1D2D3$]{
      \includegraphics[width=0.31\textwidth]{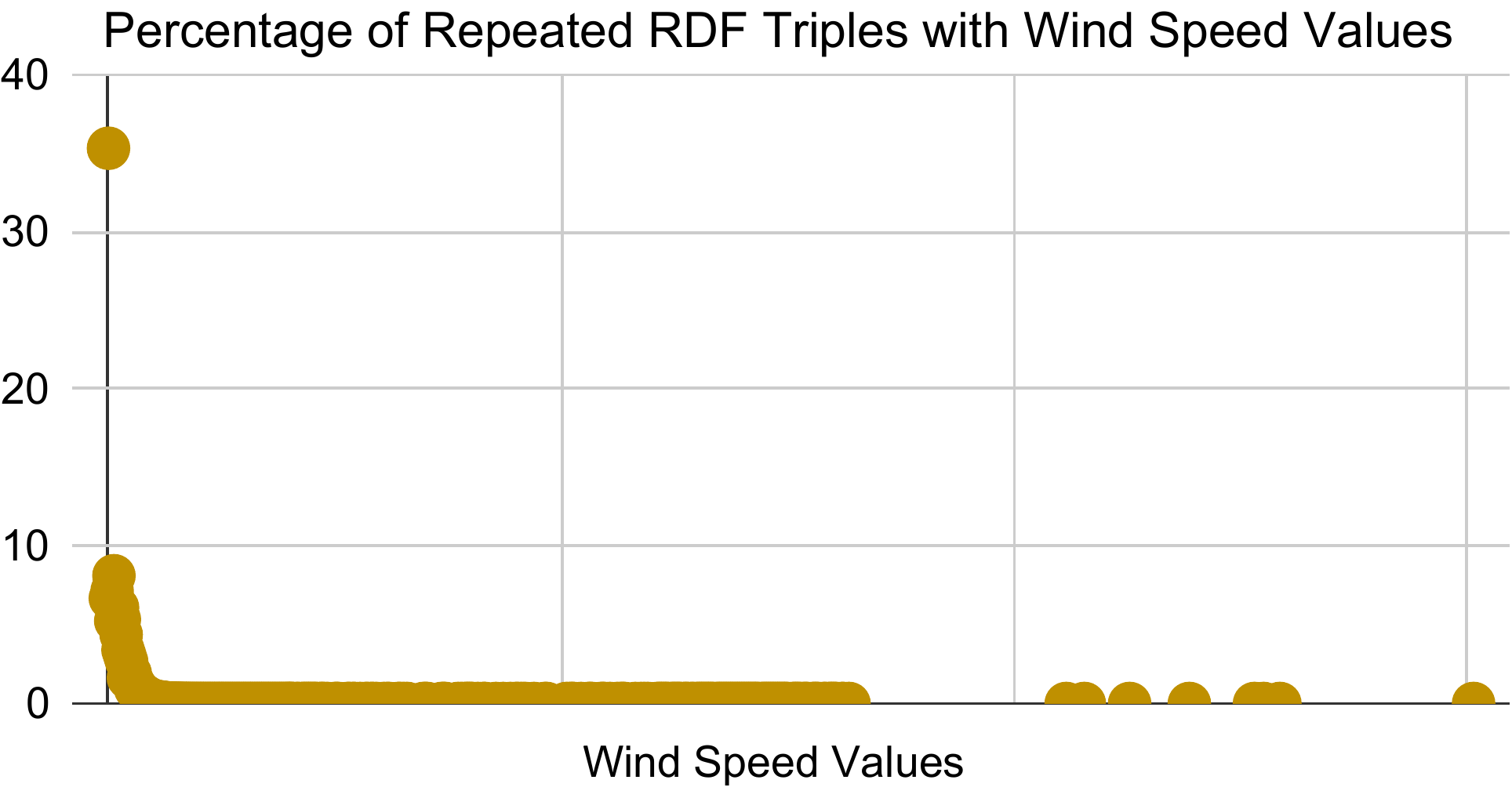}
      \label{fig:distWS}}
       \vspace{-0.50pt}
       \vspace{0pt}\subfloat[\%age of Temperature Repeated Triples in $D1D2D3$]{
      \includegraphics[width=0.31\textwidth]{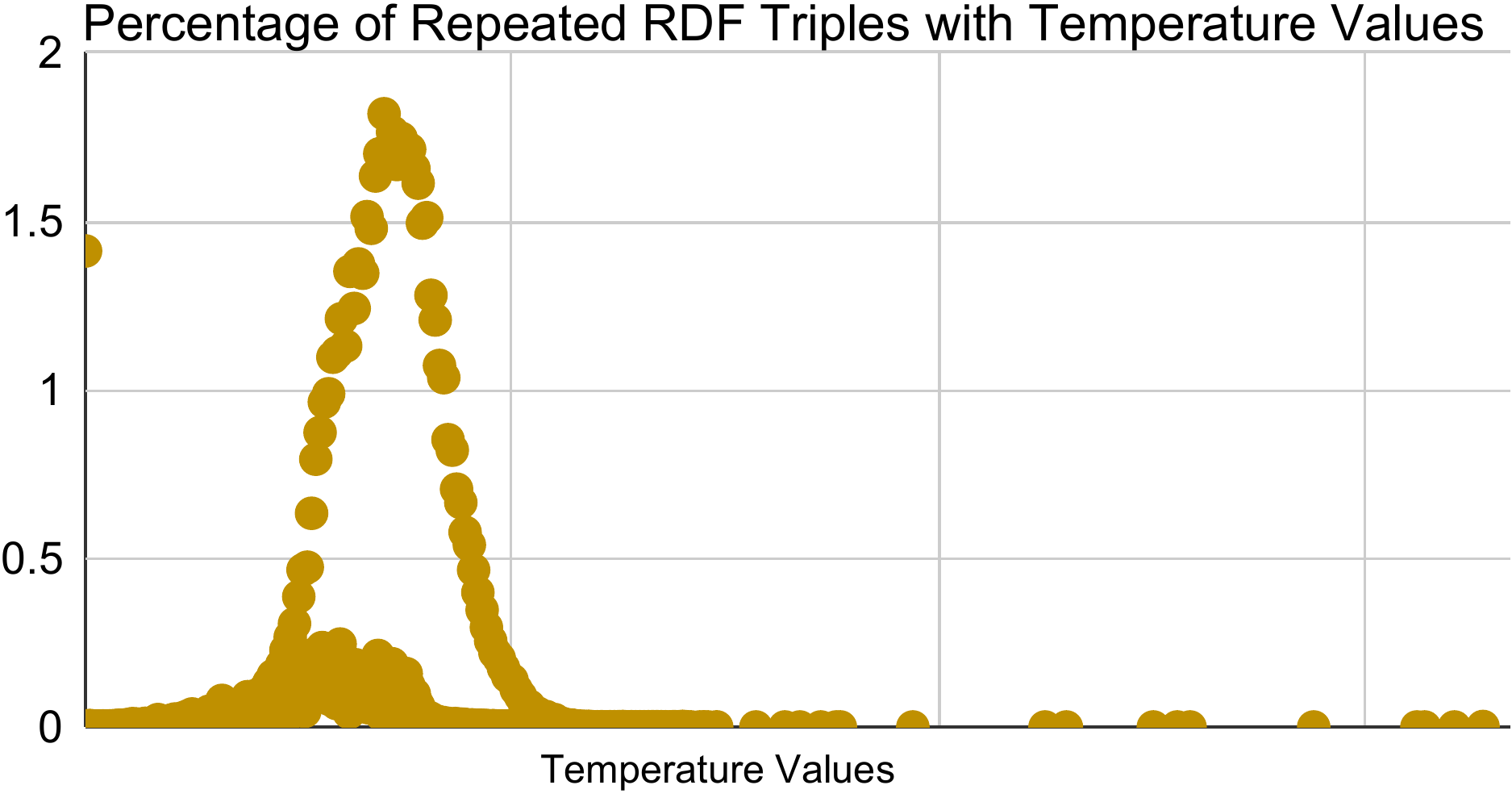}
      \label{fig:distTemp}}
       \vspace{-0.50pt}
  \subfloat[\%age of Relative Humidity Repeated Triples in $D1D2D3$]{
      \includegraphics[width=0.31\textwidth]{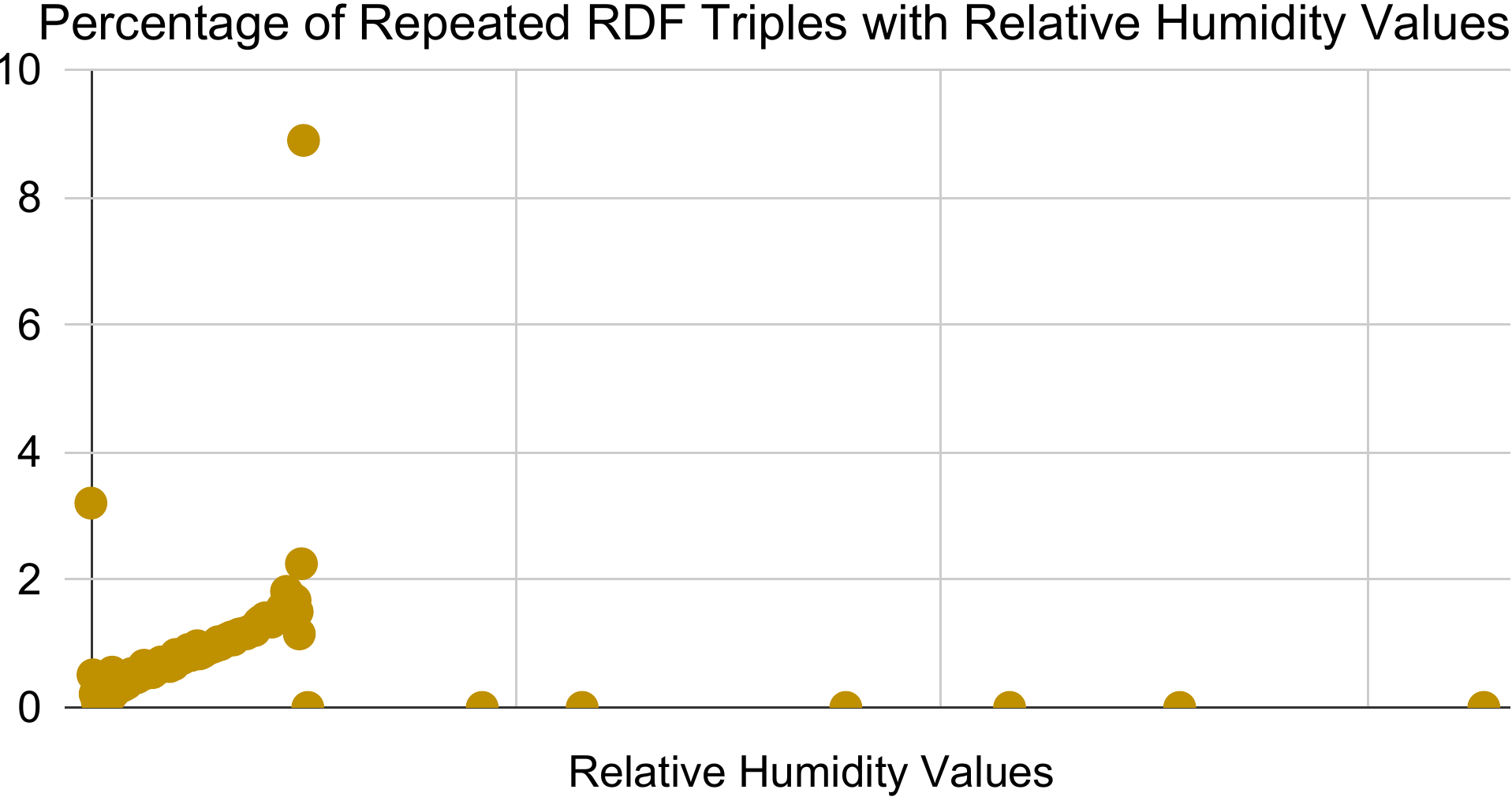}
      \label{fig:distRH}}
   \caption{{\bf Percentage of Repeated RDF Triples with Observation Values}. Few of the large number of values are highly repeated. (a) Percentage of repeated RDF triples with windspeed values; (b) Percentage of repeated triples with temperature values; (c) Percentage of repeated triples with relative humidity values.}
\label{fig:distribution}
\end{figure*} 

\noindent
\textbf{ Metrics.} We measure the results of our empirical evaluation in terms of number of nodes and edges in an RDF graph. The size of an RDF graph is presented as the sum of nodes and edges in the graph, where the nodes correspond to the entities and objects, whereas the edges are labeled edges annotated with the properties of a class in an RDF graph. In our empirical evaluation, we report on the following metrics:
\begin{inparaenum}[\bf a\upshape)]
\item {\bf Execution Time (Exec.Time(ms))} is the time required to find the frequent star patterns in RDF graphs.
\item {\bf Number of Nodes (NN)} is the number of \textit{Observation} and \textit{Measurement} entities and objects in RDF graphs.
\item {\bf Number of Labeled Edges (NLE)} represents the number of labeled edges annotated with the properties in \textit{Observation} and \textit{Measurement} classes in RDF graphs.
\item {\bf Percentage Savings in the Number of Labeled Edges (\%Savings)} stands for the percentage increase or decrease in the number of labeled edges using a positive or a negative value, respectively. 
%The labeled edges relate the entities and the objects in \textit{Observation} and \textit{Measurement} classes. 
The interpretation of the metric \textbf{\%Savings} is, higher is better.
\end{inparaenum}

\noindent
\textbf{Implementation.}
The experiments were performed on a Linux Debian 8 machine with a CPU Intel Xeon(R) Platinum 8160 2.10GHz and 754GB RAM.
The datasets are factorized for \textit{Observation} and \textit{Measurement} classes using all possible combinations of the properties in each class.
Table~\ref{tab:attributes} shows the set of properties for \textit{Observation} and \textit{Measurement} (Meas.), respectively, in the SSN ontology. Each set of properties is assigned a set identification string \textit{SID}, and are referred with the corresponding identification string in the paper. 
\textit{Observation} contains \textit{property}, \textit{procedure}, \textit{generatedBy}, and \textit{time} property. \textit{procedure} and \textit{generatedBy} are symmetric properties and are considered together in the sets. Similarly, in \textit{Measurement}, sets of properties contain the properties \textit{value} and \textit{unit}.
Further, for experiments, datasets are gradually merged to increase datasets size. The source code is available at https://github.com/SDM-TIB/Graph-Factorization.

\subsection{Efficiency of Frequent Star Patterns Detection Approach}
For evaluating the efficiency of the proposed frequent star patterns techniques and to answer the research question {\bf RQ1}, we execute \emph{E.FSP} and \emph{G.FSP} over five percent of RDF triples from dataset $D1$. The dataset of the selected RDF triples describe the \textit{Measurement} and \textit{Observation} classes, where several different types of observations from the \textit{Observation} class are included in the dataset. gSpan~\cite{yan2002gspan} is used to generate the frequent patterns space for \emph{E.FSP}, which iterates over all the generated frequent patterns. To evaluate the efficiency of two approaches, we selected five percent of RDF triples from the dataset $D1$; this number was chosen as a timeout because gSpan was able to generate the frequent patterns within thirty minutes. Efficiency comparison in terms of execution time of \emph{E.FSP} and \emph{G.FSP} is reported in Table~\ref{tab:gspanVSfsp}. 
\emph{G.FSP} finds the frequent star patterns without generating all the star patterns involving all the possible subsets of properties.
Table~\ref{tab:gspanVSfsp} shows for \emph{E.FSP} and \emph{G.FSP}, the number of iterations over sets of properties \textit{PSIterations}, the number of frequent star patterns detected \textit{\#FSP}, and the execution time in milliseconds \textit{Exec.Time(ms)} required to detect the frequent star patterns. The results indicate that \emph{E.FSP} and \emph{G.FSP} detect the same frequent star patterns. The frequent star patterns, detected by \emph{E.FSP} and \emph{G.FSP}, are over the set of properties $A5$ and $A8$ for all the different observations in the \textit{Observation} class, and the \textit{Measurement} class, respectively. Execution time of  \emph{G.FSP} to detect frequent star patterns is less by at least three orders of magnitude than the execution time of \emph{E.FSP}, e.g., \emph{G.FSP} detects frequent star patterns in measurement class in \num{1.9e2} milliseconds, whereas \num{5.3e5} milliseconds are required using \emph{E.FSP}. 

\begin{table*}[b!]
   \caption{{\bf Efficiency of Frequent Star Patterns Detection.} \emph{E.FSP} and \emph{G.FSP} are used to detect the frequent star patterns for the \textit{Observation} and \textit{Measurement} classes in the five percent of RDF triples from the dataset $D1$. \emph{E.FSP} and \emph{G.FSP} detect the same frequent star patterns involving the sets $A5$ and $A8$ of properties from the \textit{Observation} and \textit{Measurement} classes, respectively. \emph{G.FSP} takes less time to identify the same frequent star patterns than the time taken by \emph{E.FSP}.}
  \centering
  \subfloat{
  \resizebox{\textwidth}{!}{%
  \begin{tabular}{|llllllll|}
 \cline{3-8}
  \multicolumn{1}{c}{}&\multicolumn{1}{c}{}  & \multicolumn{2}{|c}{\bf{PSIterations}} &  \multicolumn{2}{|c}{\bf{\#FSP}} & \multicolumn{2}{|c|}{\bf{Exec.Time(ms)}}  \\ \hline
  
 \multicolumn{2}{|c}{\bf{Class}} & \multicolumn{1}{|c}{\bf{E.FSP}} & \multicolumn{1}{|c}{\bf{G.FSP}}  & \multicolumn{1}{|c}{\bf{E.FSP}} & \multicolumn{1}{|c}{\bf{G.FSP}} & \multicolumn{1}{|c}{\bf{E.FSP}} & \multicolumn{1}{|c|}{\bf{G.FSP}}  \\ \hline
  
  \multirow{9}{*}{\rotatebox[origin=c]{90}{Observation}} &\multicolumn{1}{|c}{Precipitation} &    \multicolumn{1}{|r}{8} & \multicolumn{1}{|r}{4}   &  \multicolumn{1}{|r}{23}   & \multicolumn{1}{|r}{23}&  \multicolumn{1}{|r}{\num{2.1e4}}   & \multicolumn{1}{|r|}{\num{1.5e1}}    \\ \cline{2-8}
 
  &\multicolumn{1}{|c}{Pressure} &    \multicolumn{1}{|r}{5} & \multicolumn{1}{|r}{4}  &  \multicolumn{1}{|r}{183}   & \multicolumn{1}{|r}{183} &  \multicolumn{1}{|r}{\num{1.3e6}}   & \multicolumn{1}{|r|}{\num{7.1e2}} \\ \cline{2-8}
 
 & \multicolumn{1}{|c}{Rainfall} &  \multicolumn{1}{|r}{5} & \multicolumn{1}{|r}{4}  &  \multicolumn{1}{|r}{533} & \multicolumn{1}{|r}{533} &  \multicolumn{1}{|r}{\num{1.3e6}} & \multicolumn{1}{|r|}{\num{8.0e2}} \\ \cline{2-8}
 
 & \multicolumn{1}{|c}{RelativeHumidity} &  \multicolumn{1}{|r}{5} & \multicolumn{1}{|r}{4}  & \multicolumn{1}{|r}{341} & \multicolumn{1}{|r}{341} & \multicolumn{1}{|r}{\num{1.3e6}} & \multicolumn{1}{|r|}{\num{7.5e2}} \\ \cline{2-8}
 
  &\multicolumn{1}{|c}{Snowfall} &  \multicolumn{1}{|r}{8} & \multicolumn{1}{|r}{4}  & \multicolumn{1}{|r}{382}  & \multicolumn{1}{|r}{382} & \multicolumn{1}{|r}{\num{9.2e5}}  & \multicolumn{1}{|r|}{\num{3.1e2}} \\ \cline{2-8}
 
  &\multicolumn{1}{|c}{Temperature} &  \multicolumn{1}{|r}{5} & \multicolumn{1}{|r}{4}   & \multicolumn{1}{|r}{395}  & \multicolumn{1}{|r}{395} & \multicolumn{1}{|r}{\num{1.3e6}}  & \multicolumn{1}{|r|}{\num{7.8e2}} \\ \cline{2-8}
  
   & \multicolumn{1}{|c}{Visibility} &  \multicolumn{1}{|r}{5} & \multicolumn{1}{|r}{4}  & \multicolumn{1}{|r}{395}  & \multicolumn{1}{|r}{395} & \multicolumn{1}{|r}{\num{1.3e6}}  & \multicolumn{1}{|r|}{\num{7.3e2}} \\ \cline{2-8}
    
    & \multicolumn{1}{|c}{WindDirection} &  \multicolumn{1}{|r}{5} & \multicolumn{1}{|r}{4} & \multicolumn{1}{|r}{350}  & \multicolumn{1}{|r}{350} & \multicolumn{1}{|r}{\num{1.3e6}}  & \multicolumn{1}{|r|}{\num{7.5e2}} \\ \cline{2-8}
     
     & \multicolumn{1}{|c}{WindSpeed} &  \multicolumn{1}{|r}{5} & \multicolumn{1}{|r}{4} & \multicolumn{1}{|r}{410}  & \multicolumn{1}{|r}{410}  & \multicolumn{1}{|r}{\num{1.3e6}}  & \multicolumn{1}{|r|}{\num{7.6e2}} \\ \hline
     
   \multicolumn{2}{|c}{Measurement} & \multicolumn{1}{|r}{1} & \multicolumn{1}{|r}{1}  & \multicolumn{1}{|r}{1,907} & \multicolumn{1}{|r}{1,907} & \multicolumn{1}{|r}{\num{5.3e5}} & \multicolumn{1}{|r|}{\num{1.9e2}} \\ \hline

  \end{tabular}}
 }
  \label{tab:gspanVSfsp}
\end{table*}

\subsection{Effectiveness of Frequent Star Patterns Detection Approach}
To answer the research questions {\bf RQ2} and {\bf RQ3},  we compute the values of Formula~\ref{eq:minProblem} for all the sets of properties given in Table~\ref{tab:attributes} for the \textit{Observation} and \textit{Measurement} classes, respectively, in the three RDF datasets. The computed formula values for the \textit{Observation} and \textit{Measurement} classes are shown in Table~\ref{tab:formulaValue}. Moreover, we compute the size of the original and factorized RDF graphs, in terms of nodes and edges in the RDF graphs. The formula values are computed for the sets of properties that contain only one property in the set, as well as the factorization is performed using these sets of properties to illustrate the association between the formula values and the savings obtained in the factorized graphs. Table~\ref{tab:formulaValue} shows that the set $A5$ of properties in the \textit{Observation} class generates the smaller values $D1=4,142,727$, $D1D2=15,756,888$, and $D1D2D3=334,898,603$ for the Formula~\ref{eq:minProblem}, than all the other sets $A1$, $A2$, $A3$, $A4$, $A6$, and $A7$. A smaller formula value for $A5$ indicates that the RDF graphs encapsulate a minimum number of star patterns, over the properties in the set $A5$ such that a large number of entities of the \textit{Observation} class match these star patterns. Therefore, replacing these star patterns with the compact RDF molecules during the factorization reduces the size of the RDF graphs. Figure~\ref{fig:ObsNNNE} presents the number of \textit{Observation} nodes $NN$ and the labeled edges $NLE$  in the original and factorized RDF datasets $D1$, $D1D2$, and $D1D2D3$. The results show that factorization of the \textit{Observation} class over the set $A5$ of properties reduces the sum of the number of observation nodes and the labeled edges in the factorized RDF graphs by up to $37\%$. 
On contrary, a large formula value for $A4$ in datasets $D1=4,142,727$, $D1D2=15,756,888$, and $D1D2D3=334,898,603$, than the other sets $A1$, $A2$, $A3$, $A5$, $A6$, and $A7$ indicates that a large number of star patterns over the properties in $A4$ exist in the RDF graphs and a small number of entities of the \textit{Observation} class match these star patterns. Figure~\ref{fig:ObsNNNE} depicts an increase in the number of \textit{Observation} nodes $NN$ and the labeled edges $NLE$  in the factorized RDF datasets $D1$, $D1D2$, and $D1D2D3$ after factorizing over the properties in $A4$. Similarly, the results for $A1$, $A2$, $A3$, $A6$, and $A7$ in Table~\ref{tab:formulaValue} and Figure~\ref{fig:ObsNNNE} clearly show that the higher the formula value for a set of properties increases the number of nodes and edges in the factorized RDF graphs by factorizing using the properties in the corresponding  set. 
In case of the \textit{Measurement} class Table~\ref{tab:formulaValue} shows smaller  formula values for the set $A8$ of properties i.e., $D1=28,491$, $D1D2=34,554$, and $D1D2D3=40,302$, than the other sets $A9$ and $A10$. %This indicates that in the three RDF graphs, the number of star patterns over all the properties of the \textit{Measurement} class in $A8$ is minimum and a large number of entities of the class map to these star patterns as compared to $A9$ and $A10$. 
Figure~\ref{fig:measNNNE} reports the sum of the nodes and the labeled edges representing measurements in the original and factorized RDF datasets $D1$, $D1D2$, and $D1D2D3$. 
The sum of the nodes and the labled edges of the measurements are reduced up to $60\%$ in all the factorized RDF graphs by factorizing over the properties in $A8$. Furthermore, the higher formula values for the sets $A9$ and $A10$ indicate less savings after factorization compared to the set $A8$. The number of nodes and edges in the factorized RDF graphs by factorizing over the properties in sets $A9$ and $A10$ in Figure~\ref{fig:measNNNE} are higher than $A8$.
These results show that the different combinations of class properties impact the factorization of RDF graphs and the proposed frequent star patterns detection techniques are able to detect the set of properties involved in the generation of frequent star patterns. Moreover, our techniques are able to anticipate the best set of properties, answering thus, research questions {\bf RQ2} and {\bf RQ3}.

\begin{table*}[b!]
   \caption{{\bf Values Computed for Formula~\ref{eq:minProblem}.} The sets of properties in Table~\ref{tab:attributes} for the \textit{Observation} and \textit{Measurement} (Meas.) classes, respectively, are used to compute the Formula~\ref{eq:minProblem} values over the RDF datasets $D1$, $D1D2$, and $D1D2D3$. The minimum formula values for the \textit{Observation} and \textit{Measurement} classes and the corresponding sets $A5$ and $A8$, respectively, of properties are highlighted in bold. The smaller formula values for $A5$ and $A8$ in the \textit{Observation} and \textit{Measurement} classes, respectively, indicate the maximum savings after factorizing the RDF graphs over the properties in $A5$ and $A8$, as shown in Figure~\ref{fig:ne} and Table~\ref{tab:NELFactorizedObsMeas}.}
  \centering
  \subfloat{
  \resizebox{0.7\textwidth}{!}{%
  \begin{tabular}{|lllll|}
  \cline{3-5}
\multicolumn{2}{c}{}  & \multicolumn{3}{|c|}{$\mathbf{\#Edges(SP,C,G)}$}  \\ \cline{2-5}

  \multicolumn{1}{c}{} & \multicolumn{1}{|c}{\bf{SID}} & \multicolumn{1}{|C{2cm}}{\bf{D1}} &  \multicolumn{1}{|C{2cm}}{\bf{D1D2}} & \multicolumn{1}{|c|}{\bf{D1D2D3}} \\ \hline
  
\multirow{7}{*}{\rotatebox[origin=c]{90}{\bf{Observation}}}& \multicolumn{1}{|c}{A1} & \multicolumn{1}{|r}{12,071,185} & \multicolumn{1}{|r}{46,643,440} & \multicolumn{1}{|r|}{103,815,183} \\ \cline{2-5} 

& \multicolumn{1}{|c}{A2} & \multicolumn{1}{|r}{12,090,195} & \multicolumn{1}{|r}{46,687,690} & \multicolumn{1}{|r|}{103,891,717} \\ \cline{2-5}

& \multicolumn{1}{|c}{A3} & \multicolumn{1}{|r}{8,111,623} & \multicolumn{1}{|r}{31,205,888} & \multicolumn{1}{|r|}{69,358,875} \\ \cline{2-5}

& \multicolumn{1}{|c}{A4} & \multicolumn{1}{|r}{20,118,595} & \multicolumn{1}{|r}{78,698,580} & \multicolumn{1}{|r|}{174,865,870} \\ \cline{2-5}

& \multicolumn{1}{|c}{\bf{A5}} & \multicolumn{1}{|r}{\bf{4,142,727}} & \multicolumn{1}{|r}{\bf{15,756,888}} & \multicolumn{1}{|r|}{\bf{34,898,603}} \\ \cline{2-5}

& \multicolumn{1}{|c}{A6} & \multicolumn{1}{|r}{8,097,964} & \multicolumn{1}{|r}{31,245,605} & \multicolumn{1}{|r|}{69,474,786} \\ \cline{2-5}

& \multicolumn{1}{|c}{A7} & \multicolumn{1}{|r}{15,784,707} & \multicolumn{1}{|r}{61,406,644} & \multicolumn{1}{|r|}{135,902,747} \\ \hline

\multirow{3}{*}{\rotatebox[origin=c]{90}{\bf{Meas.}}}& \multicolumn{1}{|c}{\bf{A8}} & \multicolumn{1}{|r}{\bf{28,491}} & \multicolumn{1}{|r}{\bf{34,554}} & \multicolumn{1}{|r|}{\bf{40,302}}\\ \cline{2-5}

& \multicolumn{1}{|c}{A9} & \multicolumn{1}{|r}{4,037,067} & \multicolumn{1}{|r}{15,563,838} & \multicolumn{1}{|r|}{34,623,579} \\ \cline{2-5}

& \multicolumn{1}{|c}{A10} & \multicolumn{1}{|r}{4,023,731} & \multicolumn{1}{|r}{15,547,816} & \multicolumn{1}{|r|}{34,605,063} \\ \hline
 \end{tabular}}

 }
  \label{tab:formulaValue}
\end{table*}

\begin{figure*}[h!]
\centering
     \subfloat[\# of Observation Nodes $NN$ and Edges $NLE$]{
      \includegraphics[width=.5\textwidth]{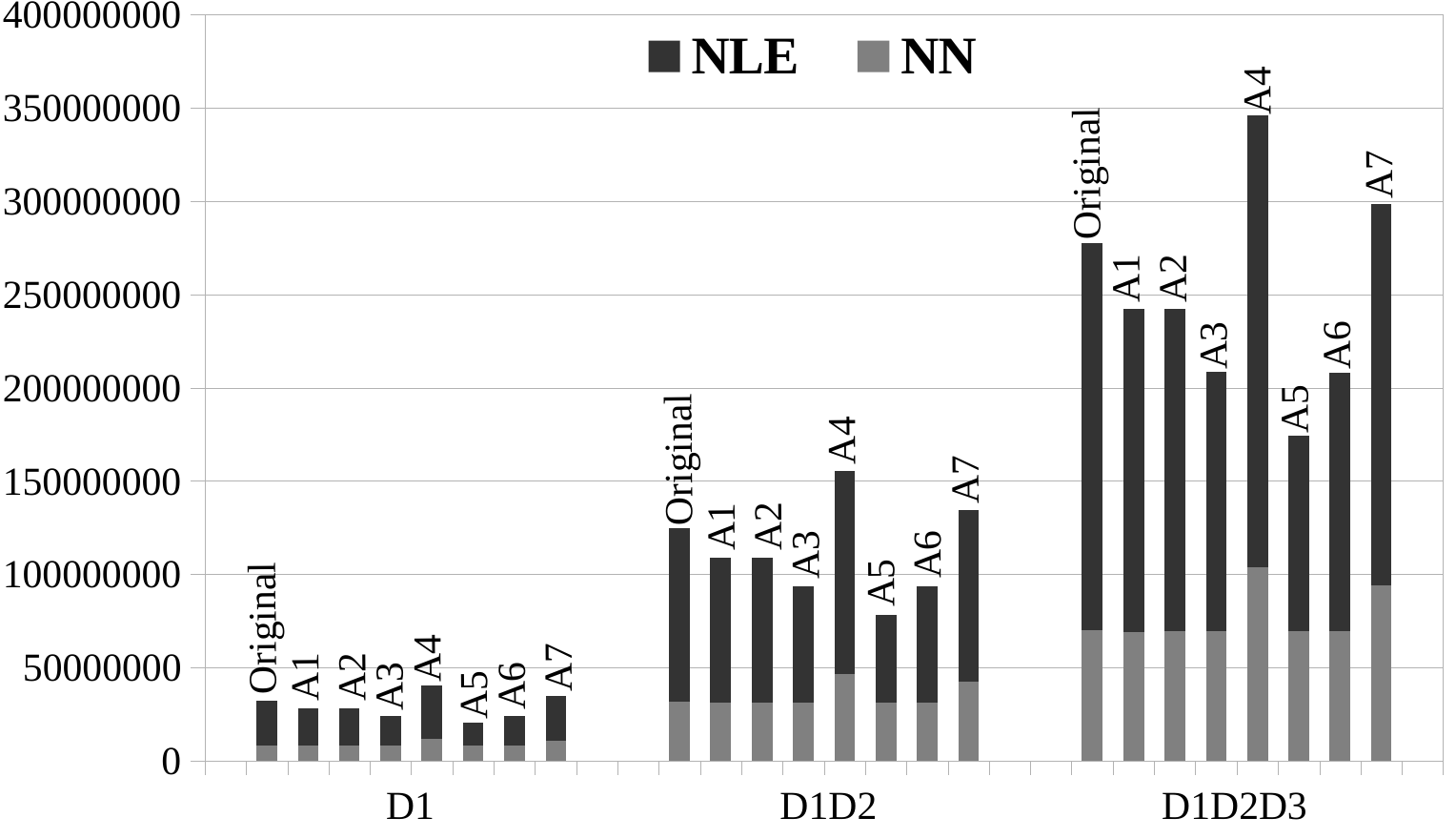}
      \label{fig:ObsNNNE}}
    \subfloat[\# of Measurement nodes $NN$ and edges $NLE$]{
      \includegraphics[width=0.51\textwidth]{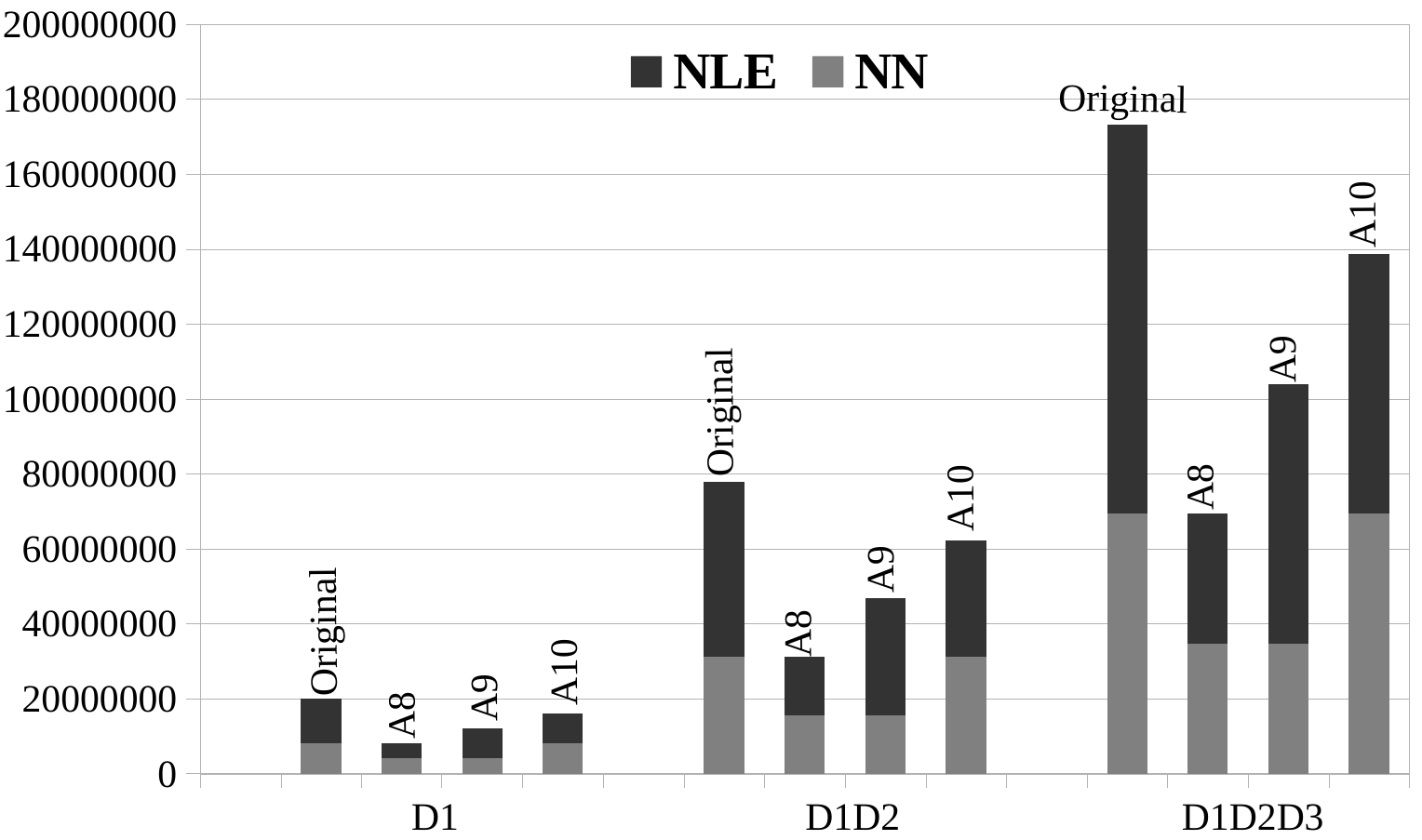}
      \label{fig:measNNNE}}
   \caption{{\bf Nodes and Labeled Edges}. The number of nodes \textit{NN} and labeled edges \textit{NLE} before and after factorization of the RDF datasets. (a) The number of nodes $NN$ and labeled edges $NLE$ representing observations in the RDF datasets; (b) The number of nodes $NN$ and labeled edges $NLE$ representing measurements.}
\label{fig:ne}
\end{figure*}

\begin{table*}[b!]
   \caption{{\bf Percentage Savings in Labeled Edges after Factorization.} Savings \textit{\%Savings} in the number of Labeled Edges \textit{NLE($G'$)} after factorization of the RDF datasets using the sets of properties in Observation and Measurement classes.}
  \centering
  \subfloat{
  \resizebox{\textwidth}{!}{%
  \begin{tabular}{|llllllll|}
 \cline{3-8}
  \multicolumn{1}{c}{} & &\multicolumn{2}{|c}{\bf{D1}} &  \multicolumn{2}{|c}{\bf{D1D2}} & \multicolumn{2}{|c|}{\bf{D1D2D3}}  \\ \hline
 
   \multirow{7}{*}{\rotatebox[origin=c]{90}{\bf{Observation}}}&\multicolumn{1}{|c}{\bf{SID}} & \multicolumn{1}{|c}{\bf{NLE($G'$)}} & \multicolumn{1}{|c}{\bf{\%Savings}} & \multicolumn{1}{|c}{\bf{NLE($G'$)}} & \multicolumn{1}{|c}{\bf{\%Savings}} & \multicolumn{1}{|c}{\bf{NLE($G'$)}} & \multicolumn{1}{|c|}{\bf{\%Savings}}  \\ \cline{2-8}
  
& \multicolumn{1}{|c}{A1} & \multicolumn{1}{|r}{20,125,493} & \multicolumn{1}{|r}{16.64}  & \multicolumn{1}{|r}{77,745,918} & \multicolumn{1}{|r}{16.66} & \multicolumn{1}{|r}{173,032,155} & \multicolumn{1}{|r|}{16.66} \\ \cline{2-8}
 
 &  \multicolumn{1}{|c}{A2} &    \multicolumn{1}{|r}{20,144,503} & \multicolumn{1}{|r}{16.56}  &  \multicolumn{1}{|r}{77,790,168}   & \multicolumn{1}{|r}{16.61} &  \multicolumn{1}{|r}{173,108,689}   & \multicolumn{1}{|r|}{16.63}   \\ \cline{2-8}
 
  & \multicolumn{1}{|c}{A3} &    \multicolumn{1}{|r}{16,226,021} & \multicolumn{1}{|r}{32.79}  &   \multicolumn{1}{|r}{62,546,938}  & \multicolumn{1}{|r}{32.95} &  \multicolumn{1}{|r}{139,064,503}   & \multicolumn{1}{|r|}{33.02} \\ \cline{2-8}
 
   &\multicolumn{1}{|c}{A4} &  \multicolumn{1}{|r}{28,170,155} & \multicolumn{1}{|r}{-16.68} & \multicolumn{1}{|r}{108,838,750}  & \multicolumn{1}{|r}{-16.67}  &  \multicolumn{1}{|r}{242,239,479} & \multicolumn{1}{|r|}{-16.67} \\ \cline{2-8}
 
  &\multicolumn{1}{|c}{\bf A5} &  \multicolumn{1}{|r}{\bf 12,277,576} & \multicolumn{1}{|r}{\bf 49.14} & \multicolumn{1}{|r}{\bf 47,175,356}  & \multicolumn{1}{|r}{\bf 49.43} & \multicolumn{1}{|r}{\bf 104,786,128} & \multicolumn{1}{|r|}{\bf 49.53} \\ \cline{2-8}
 
  &\multicolumn{1}{|c}{A6} &  \multicolumn{1}{|r}{16,150,898} & \multicolumn{1}{|r}{33.10}  & \multicolumn{1}{|r}{62,317,489}  & \multicolumn{1}{|r}{33.20} & \multicolumn{1}{|r}{138,639,234}  & \multicolumn{1}{|r|}{33.23} \\ \cline{2-8}
 
  &\multicolumn{1}{|c}{A7} &  \multicolumn{1}{|r}{23,837,352} & \multicolumn{1}{|r}{1.26}  &  \multicolumn{1}{|r}{92,088,523} & \multicolumn{1}{|r}{1.28}  & \multicolumn{1}{|r}{204,304,156}  & \multicolumn{1}{|r|}{1.60} \\ \hline
  
  \multirow{3}{*}{\rotatebox[origin=c]{90}{\bf{Meas.}}}& \multicolumn{1}{|c}{\bf A8} & \multicolumn{1}{|r}{\bf 4,059,738} & \multicolumn{1}{|r}{\bf 66.37} & \multicolumn{1}{|r}{\bf 15,599,469} & \multicolumn{1}{|r}{\bf 66.56} & \multicolumn{1}{|r}{\bf 34,716,176} & \multicolumn{1}{|r|}{\bf 66.56} \\ \cline{2-8}
 
 & \multicolumn{1}{|c}{A9} &  \multicolumn{1}{|r}{8,069,688}  & \multicolumn{1}{|r}{33.15} &  \multicolumn{1}{|r}{31,130,127}  & \multicolumn{1}{|r}{33.26} &  \multicolumn{1}{|r}{69,300,827}  & \multicolumn{1}{|r|}{33.25}  \\ \cline{2-8}
 
 & \multicolumn{1}{|c}{A10} &  \multicolumn{1}{|r}{8,056,352} & \multicolumn{1}{|r}{33.26} & \multicolumn{1}{|r}{31,114,105}  & \multicolumn{1}{|r}{33.29} & \multicolumn{1}{|r}{69,282,311}  & \multicolumn{1}{|r|}{33.26}  \\ \hline
  
 \end{tabular}}
  \label{tab:NELFactorizedObs}}

  \label{tab:NELFactorizedObsMeas}
\end{table*}

%measurement class

\subsection{Effectiveness of RDF  Graph Factorization}
We factorize the gradually increasing RDF datasets  $D1$, $D1D2$, and $D1D2D3$ over the  \textit{Observation} and \textit{Measurement} classes using the properties in the sets of properties given in Table~\ref{tab:attributes}.
The percentage savings are computed in terms of labeled edges for the observations and measurements in the RDF datasets after factorization.
Table~\ref{tab:NELOriginal} presents the number of edges $NLE(G)$ in the \textit{Observation} and \textit{Measurement} classes in the original RDF datasets $D1$, $D1D2$, and $D1D2D3$.
Table~\ref{tab:NELFactorizedObsMeas} presents the number of labeled edges $NLE(G')$ and the percentage savings $\%savings$ after factorization of the \textit{Observation} and \textit{Measurement} classes.
The highest savings $49.14\%$, $49.43\%$, and $49.53\%$ in $NLE(G')$ for observations after factorizing $D1$, $D1D2$, and $D1D2D3$ over the properties in $A5$, shows that the number of frequent star patterns over the properties in $A5$ are reduce by replacing them with the corresponding compact RDF molecules.
On the other hand, the set $A4$ of properties gives negative values of percentage savings $\%Savings$, $-16.68\%$, for the RDF dataset $D1$, and $-16.67\%$, for the RDF datasets $D1D2$ and $D1D2D3$, indicating an increase in the number of labeled edges after the factorization of the RDF datasets.
Similarly, for measurements, the positive values $66.37\%$ of percentage savings after factorizing $D1$, and $66.56\%$ for $D1D2$ and $D1D2D3$ over $A8$ indicate a decrease in the number of labeled edges after factorization. Furthermore, the percentage savings in the set $A8$ of properties are higher than in $A9$ and $A10$.
These results allow us to positively answer research question {\bf RQ4}.

\section{Conclusions and Future Work}
\label{sec:conclusion}
This article presents computational methods to identify frequent star patterns and to generate a \emph{factorized RDF graph}, with a minimized number of frequent star patterns. A frequent star pattern contains class entities linked to the objects or other resources using labeled edges annotated with properties in the class. These frequent star patterns introduce redundancy in terms of edges and nodes. Our proposed computational methods implement the frequent star pattern detection algorithm based on search space pruning techniques to identify the classes and properties involved in frequent star patterns. Furthermore, the proposed factorization techniques generate compact representation of RDF graphs, \emph{factorized RDF graph}, by replacing a frequent star pattern with a compact RDF molecule, composed of a surrogate entity connected to the object in the frequent star pattern using the labeled edges annotated with relevant properties.
We empirically study the effectiveness of the frequent star pattern detection algorithm to identify class and properties involved in the frequent star pattern. Furthermore, we evaluate the impact of the factorization techniques over the gradually increasing RDF graphs size and different combinations of class properties. Experimental results suggest that the proposed computational methods successfully identify the class properties involved in the frequent star patterns and remove redundancy caused by these frequent star patterns. For the best set of properties, identified by the frequent star pattern detection algorithm, the RDF graph size is reduced by up to $66.56\%$.
%Our work broadens the repertoire of techniques available to the field
Our work broadens the repertoire of techniques for representing and storing knowledge graphs by providing RDF graph compression techniques which exploit the semantics encoded in the data; these techniques  generate compact representations of RDF graphs to help improving query processing over RDF graphs without requiring a customized engine.
%What is your hope with this article? Eg. We hope that our method will help scientists
Our work contributes to the crucial knowledge graph representation and provides the basics for further development of the efficient processing techniques over the compact knowledge graphs.
In future, we will exploit parallel processing to efficiently find frequent star patterns. 
 
\section*{Acknowledgments}
Farah Karim is supported by the German Academic Exchange Service (DAAD); this work is partially funded by the EU H2020 project IASiS (GA No.727658).

\bibliographystyle{abbrv}
\bibliography{sigproc} 

\end{document}